\documentclass{scrartcl}

\usepackage{hyperref}
\usepackage{amsmath,amssymb}
\usepackage{amsthm}
\usepackage[utf8]{inputenc}
\usepackage{bbm}
\usepackage{authblk}

\usepackage{tikz}
\usetikzlibrary{patterns}

\newtheorem{theorem}{Theorem}[section]
\newtheorem{lemma}[theorem]{Lemma} 
\newtheorem{remark}[theorem]{Remark}
\newtheorem{definition}[theorem]{Definition}
\newtheorem{proposition}[theorem]{Proposition}

\newcommand{\diag}{\text{diag}}
\newcommand{\trace}{\text{Tr}\,}

\newcommand{\C}{\mathbb{C}}
\newcommand{\Id}{\text{Id}}
\newcommand{\E}{\mathbb{E}}
\newcommand{\real}{\text{Re}}

\newcommand{\Bnorm}[1]{\Vert  #1 \Vert_{B_1}}
\newcommand{\wBnorm}[1]{\Vert #1 \Vert_{B_{1, w}}}
\newcommand{\innerproduct}[1]{ \langle #1 \rangle }

\newcommand{\fmend}{f \left( X_{\ell} \right)}
\newcommand{\inff}{\underset{f \in \mathcal{F}}{\inf}}
\newcommand{\supf}{\underset{f \in \mathcal{F}}{\sup}}
\newcommand{\dcone}[1]{\mathcal{K}_{\ast}  \left(  #1 \right) }

\begin{document}
\title{On the convex geometry of blind deconvolution and matrix completion\thanks{The results of this paper have been presented in part at the 52nd Annual Asilomar Conference on Signals, Systems, and Computers, October 28-31, 2018, Pacific Grove, USA \cite{asilomar2018}}}
\author{ Felix Krahmer\thanks{Dept. of Mathematics, Technische Universit\"at M\"unchen}, Dominik St\"oger\thanks{Dept. of Electrical and Computer Engineering, University of Southern California}}
\maketitle
\thispagestyle{plain}
\pagestyle{plain}
\begin{abstract}
Low-rank matrix recovery from structured measurements has been a topic of intense study in the last decade and many important problems like matrix completion and blind deconvolution have been formulated in this framework. An important benchmark method to solve these problems is to minimize the nuclear norm, a convex proxy for the rank. A common approach to establish recovery guarantees for this convex program relies on the construction of a so-called approximate dual certificate. However, this approach provides only limited insight in various respects. Most prominently, the noise bounds exhibit seemingly suboptimal dimension factors. In this paper we take a novel, more geometric viewpoint to analyze both the matrix completion and the blind deconvolution scenario. We find that for both these applications the dimension factors in the noise bounds are not an artifact of the proof, but the problems are intrinsically badly conditioned. We show, however, that bad conditioning only arises for very small noise levels: Under mild assumptions that include many realistic noise levels we derive near-optimal error estimates for blind deconvolution under adversarial noise.
\end{abstract}


\section{Introduction}\label{section:introduction}
A number of recent works have explored the observation that various ill-posed inverse problems in signal processing, imaging, and machine learning can be naturally formulated as the task of recovering a low-rank matrix $X_0 \in \mathbb{C}^{n_1 \times n_2}$ from an underdetermined system of structured linear measurements
\begin{equation*}
y= \mathcal{A} \left(X_0 \right) + e \in \mathbb{C}^m,
\end{equation*}
where $\mathcal{A}: \mathbb{C}^{n_1 \times n_2} \rightarrow \mathbb{C}^m $ is a linear map and $e \in \mathbb{C}^m$, $\Vert e \Vert \le \tau$, represents additive noise. Such problems include, for example, matrix completion \cite{candes2009exact}, phase retrieval \cite{candes2013phaselift}, blind deconvolution \cite{ahmed2014blind}, robust PCA \cite{candes2011robust}, and demixing \cite{troppdemixing}. In this paper, we aim to analyze the worst case scenario, that is, we do not make any assumptions on the noise except for the bound on its Euclidean norm (this scenario is sometimes referred to as {\em adversarial noise}, as it allows for noise specifically designed to be most harmful in a given situation). A natural first  approach to recover $X_0$ that remains an important benchmark is to solve the semidefinite program
\begin{align*}
\text{minimize} \quad &\Vert X \Vert_{\ast}\\
\text{subject to} \quad  &\Vert y - \mathcal{A} \left(X\right) \Vert \le \tau,
\end{align*}
where $\Vert \cdot \Vert_{\ast}$ denotes the nuclear norm, i.e., the sum of the singular values. Recovery guarantees have been shown under the assumption that the measurement operator $ \mathcal{A}$ possesses a certain degree of randomness. To establish such guarantees various proof strategies have been proposed, including approaches via the restricted isometry property \cite{recht2010guaranteed,liu2011universal}, descent cone analysis \cite{chandrasekaran2012convex}, and so-called approximate dual certificates \cite{gross2010quantum,gross2011recovering}. While the latter approach remains state of the art for many structured problems including the highly relevant problems of randomized blind deconvolution and matrix completion, it seemingly has some disadvantages. 
Most prominently, the resulting recovery guarantees take the form 
\begin{equation}\label{ineq:additionalfactor}
\Vert \hat{X} - X_0 \Vert_F \lesssim \sqrt{n_1} \tau,
\end{equation}
where $\hat{X}$ denotes a minimizer of the semidefinite program above and $\Vert \cdot \Vert_F $ denotes the Frobenius norm, whereas under comparable normalization, the first two approaches, when applicable, give rise to superior recovery guarantees of the form
\begin{equation*}
\Vert \hat{X} - X_0 \Vert_F \lesssim \tau.
\end{equation*}
Before this paper it was open whether the additional dimension scaling factor in (\ref{ineq:additionalfactor}) is a proof artifact. Similarly, for randomized blind deconvolution one of the coherence terms appearing in the result was believed to arise only from the proof technique (cf.~\cite[Remark 2]{ling2017blind}).\\

Another drawback of proceeding via an approximate dual certificate is that it gives only limited insight into geometric properties of the problems such as the null-space property \cite{cohendahmendevore}, which is also an important ingredient for the study of some more efficient non-convex algorithms \cite{irlsfornasier,irlskuemmerle}.\\ 

Approaches via descent cone analysis \cite{chandrasekaran2012convex}, in contrast, provide much more geometric insight. 
The underlying idea of such approaches is to study the minimum conic singular value defined by
\begin{equation*}
\lambda_{\min} \left( \mathcal{A}, \mathcal{K}  \right) := \underset{Z \in \mathcal{K}\backslash \left\{0\right\} }{\inf} \frac{\Vert \mathcal{A} \left(Z\right)\Vert}{\Vert Z \Vert_F}
\end{equation*}
for $\mathcal{K}$ the descent cone of the underlying atomic norm -- the nuclear norm in case of low-rank matrix recovery. For a more detailed review of this approach including a precise definition of the descent cone we refer to Section \ref{section:descentcone} below. Through the study of the minimum conic singular value many superior results were obtained for low-rank recovery problems, most importantly in the context of phase retrieval \cite{kueng2016low,kueng2017low}. Furthermore, minimum conic singular values can also help to understand certain nonlinear measurement models \cite{plannonlinear}. \\

For all these reasons, it would be desirable to apply this approach also for matrix completion and blind deconvolution. A challenge that one faces, however, is that for both problems one cannot hope to recover all low-rank matrices; rather, only matrices that satisfy certain coherence constraints are admissible (cf.~the discussion in \cite[Section 5.4]{tropp2015convex}). In this article we address this challenge, providing the first geometric analysis of these problems. We find that the dimensional factors appearing in the error bounds are the true scaling of the minimum conic singular value and hence intrinsically relate to the underlying geometry. Nevertheless for blind deconvolution, near-optimal recovery is possible, if the noise level is not too small.

\subsection{Organization of the paper and our contribution}
In Section \ref{section:background} we will review blind deconvolution, matrix completion, as well as some techniques related to descent cone analysis. In Section \ref{section:mainresults} we will present the main results of this paper. Theorems~\ref{thm:mainresult} and \ref{thm:mainresultMC} establish that for both blind deconvolution and matrix completion, nuclear norm minimization is intrinsically ill-conditioned.
In contrast, Theorem~\ref{thm:stabilityBD} provides a near-optimal error bound for blind deconvolution when the noise level is not too small, implying that the conditioning problems only take effect for very small noise levels. The upper bounds for the minimum conic singular value which are the main ingredients of Theorems~\ref{thm:mainresult} and \ref{thm:mainresultMC} are derived in Section \ref{section:proofconic}. In Section \ref{section:stability} we prove the stability results for blind deconvolution.\\
We believe that not only our results, but also the proof techniques and geometric insights in this manuscript will be of general interest and help to obtain further understanding of low-rank matrix recovery models, in particular under coherence constraints. We discuss interesting directions for future research in Section \ref{section:outlook}.

\section{Background and related work}\label{section:background}

\subsection{Blind deconvolution}\label{section:backgroundBD}
Blind deconvolution problems arise in a number of different areas in science and engineering such as astronomy, imaging, and communications. The goal is to recover both an unknown signal and an unknown kernel from their convolution. In this paper we work with the circular convolution, which is defined by
\begin{equation*}
w\ast x:= \left(\sum_{j=1}^L w_j x_{k-j} \right)_{k=1}^L,
\end{equation*}
where the index difference $k-j$ is considered modulo $L$. Without further assumptions on $w$ and $x$ this bilinear map is far from injective. Consequently, it is crucial to impose structural constraints on both $w$ and $x$. Arguably, the simplest such model is given by linear constraints, that is, both $w$ and $x$ are constrained to known subspaces.  
Such a model is reasonable in many applications. In wireless communication, for example, it makes sense to assume that the channel behaviour is dominated by the most direct paths and for the signal $x$ a subspace model can be enforced by embedding the message via a suitable coding map into a higher-dimensional space before transmission.

 The first rigorous recovery guarantees for such a model were derived by Ahmed, Recht, and Romberg \cite{ahmed2014blind}. More precisely, they assume that $w=Bh$, where $B\in \mathbb{C}^{L\times K}$ is a fixed, deterministic matrix such that $B^*B=Id_{K}$ (i.e., $B$ is an isometry) and they model $ x=C\overline{m_0} $, where $\overline{m}_0$ denotes the complex-conjugate of $m_0$. Here, the matrix $C\in \mathbb{C}^{L \times K}$ is a random matrix, whose entries are independent and identically distributed with circular symmetric normal distribution $ \mathcal{CN} \left(0, \frac{1}{\sqrt{L}} \right) $. In this paper we also adopt this model. 
 
Using the well-known fact that the Fourier transform diagonalizes the circular convolution one can rewrite 
\begin{equation*}
 w \ast x = \sqrt{L} F^* \diag \left(Fw\right) Fx,
\end{equation*}
where $F \in \mathbb{C}^{L \times L}$ denotes the normalized, unitary discrete Fourier matrix, and
\begin{equation*}
 \widehat{w \ast x} := F \left( w\ast x \right)  = \sqrt{L}  \diag\left(FBh_0\right)FC\overline{m}_0.
\end{equation*}
Denoting by $b_{\ell}$ the $\ell$th row of the matrix $\overline{FB}$, and by $c_{\ell}$ the $\ell$th row of the matrix $\sqrt{L} FC$, one observes that
\begin{equation*}
\left( \widehat{w \ast x} \right)_{\ell}= b^*_{\ell} h_0 m^*_0 c_{\ell} = \trace \left( h_0 m^*_0 c_{\ell} b^*_{\ell} \right) = \innerproduct{ b_{\ell} c^*_{\ell}, h_0 m^*_0 }_F.
\end{equation*}
Furthermore, because of the rotation invariance of the circular symmetric normal distribution all the entries of the vectors $\left\{ c_{\ell} \right\}_{\ell=1}^L $ are (jointly) independent and identically distributed with distribution $ \mathcal{CN} \left(0,1\right) $. Noting that the expression $\innerproduct{h_0 m^*_0, b_{\ell} c^*_{\ell} }_F $ is linear in $h_0 m^*_0$, Ahmed, Recht, and Romberg \cite{ahmed2014blind}  defined the operator $ \mathcal{A}: \mathbb{C}^{K \times N} \rightarrow \mathbb{C}^L $ by
\begin{equation}\label{def:Ablinddeconv}
\left(\mathcal{A} \left(X\right) \right) \left(\ell\right) := \innerproduct{b_{\ell} c_{\ell}^*, X}_F
\end{equation}
obtaining the measurement model
\begin{equation*}\label{def:yBD}
y=\widehat{w \ast x} +e= \mathcal{A} \left(X_0\right) + e,
\end{equation*}
where $e \in \mathbb{C}^L $ is additive noise and $X_0 = h_0 m^*_0 $. The goal is then to determine $h_0$ and $m_0$ from $y\in \mathbb{C}^{L}$ up to the inherent scaling ambiguity, or, equivalently, to find the rank-one matrix $X_0= h_0 m^*_0$.\\

For $e=0$, among all solutions of the equation $y= \mathcal{A} \left(X_0 \right) $, the matrix $X_0 $ is the one with the smallest rank. For this reason, Ahmed, Recht, and Romberg \cite{ahmed2014blind} suggested minimizing a natural proxy for the rank, the nuclear norm $\Vert \cdot \Vert_{\ast} $, defined as the sum of the singular values of a matrix.
\begin{equation}\label{opt:SDP}
\begin{split}
\text{minimize } \quad  &\Vert X \Vert_{\ast} \\
\text{subject to} \quad & \Vert  \mathcal{A} \left(X\right) -y  \Vert \le \tau .
\end{split}
\end{equation}
Here $ \tau >0 $ is an a priori bound for the noise level, that is, we assume that $\Vert e \Vert \le \tau $. 
For this semidefinite program, they establish the following recovery guarantee.
\begin{theorem}[\cite{ahmed2014blind}]\label{thm:ahmed}
Consider measurements of the form $y= \mathcal{A} \left(h_0 m^*_0\right) +e $ for $h_0 \in \mathbb{C}^K$, $m_0 \in \mathbb{C}^N $, $ e \in \mathbb{C}^L$,  and $\mathcal{A}$ as defined in (\ref{def:Ablinddeconv}). Assume that $\Vert e \Vert \le \tau $ and
	\begin{equation*}
	L/\log^3 L \gtrsim K\mu_{\max}^2 + N \max \left\{ \mu^2_{h_0};  \tilde{\mu}^2_{h_0}  \right\} .
	\end{equation*}
	Then with probability exceeding $ 1- \mathcal{O} \left( L^{-1}  \right) $ every minimizer $ \hat{X} $ of the SDP (\ref{opt:SDP}) satisfies
	\begin{equation}\label{ineq:noisebound}
	\Vert \hat{X} - h_0 m^*_0 \Vert_F \lesssim   \sqrt{K+N}   \tau.
	\end{equation}
\end{theorem}
Here $ \mu^2_{\max} $ and $ \mu^2_{h_0}$ are coherence parameters, which are defined via 
\begin{equation*}
\mu^2_{\max}:= \frac{L}{K} \underset{\ell \in \left[L\right]}{\max} \ \Vert b_{\ell} \Vert^2.
\end{equation*}
and
\begin{equation*}
\mu^2_{h_0}:= \frac{ L}{\Vert h_0 \Vert^2} \underset{\ell \in \left[L\right]}{\max} \ \vert \innerproduct{b_{\ell}, h_0} \vert^2 .
\end{equation*} 
The third coherence factor $\tilde{\mu}_{h_0}$ is a technical term corresponding to a partition that is constructed as a part of the proof of Theorem \ref{thm:ahmed}, which is based on the Golfing Scheme \cite{gross2011recovering}.\\

To put the impact of the coherence factors into perspective, observe that if all vectors $b_{\ell}$ have the same $ \ell_2$-norm, one obtains that $ \mu_{\max}=1 $; this will be the case, for example, when $B$ is a low-frequency Fourier matrix, as it appears for applications in wireless communication. The second coherence factor always satisfies $1\le \mu^2_{h_0} \le K \mu^2_{\max}$. If $ \mu_{h_0} $ is smaller, this indicates that the mass $\Vert h_0 \Vert^2 = \sum_{\ell=1}^{L} \vert \innerproduct{b_{\ell}, h_0} \vert^2 $ is distributed fairly evenly among $\vert \innerproduct{b_{\ell},h_0} \vert $. For example, if $ \mu_{h_0} =1 $, then $\vert \innerproduct{b_{\ell}, h_0} \vert = \frac{1}{\sqrt{L}} \Vert h_0 \Vert $ for all $\ell \in \left[L\right] $. Numerical simulations in \cite{ahmed2014blind} confirm that many $h_0$ corresponding to large $\mu_{h_0}$ show worse performance, indicating that this factor may be necessary.\\
The last coherence factor $ \tilde{\mu}_{h_0}$, in contrast, will no longer appear in our result below, which is why we refrain from detailed discussion.  We refer the interested reader to \cite[Remark 2.1]{ling2017blind} and \cite[Section 2.3]{jung2018blind} for details.\\
For generic $h_0$ the parameters $\mu_{h_0} $ and $ \tilde{\mu}_{h_0}$ are reasonably small. For example, if $h_0$ is chosen from the uniform distribution on the sphere, one can show that with high probability $\mu_{h_0} = \mathcal{O} \left( \sqrt{\log L} \right) $.\\  
For the noiseless case, i.e., $\tau=0$, Theorem \ref{thm:ahmed} yields exact recovery, and the required sample complexity $ L/\log^3 L \gtrsim K+N$ is optimal up to logarithmic factors, as the number of degrees of freedom is $K+N-1 $ (see \cite{kech2017optimal} for an exact identifiability analysis based on algebraic geometry.) However, if there is noise, the bound for the reconstruction error scales with  $\sqrt{K+N}$, in contrast to other measurement scenarios such as low-rank matrix recovery from Gaussian measurements (see, e.g., \cite{chandrasekaran2012convex}).\\ 

Let us comment on some related work. The foundational paper \cite{ahmed2014blind} has triggered a number of follow-up works on the problem of randomized blind deconvolution. A first line of works extended the result to recovering signals from their superposition $\sum_{i=1}^{r} w_i \ast x_i $, a problem often referred to as blind demixing \cite{ling2017blind,jung2018blind}. Another line of works investigated non-convex (gradient-descent based) algorithms \cite{li2018rapid,ma2017implicit,huang2018blind}, which have the advantage that they are computationally less expensive, as they operate in the natural parameter space. It has been shown that they require a near-optimal number of measurements for recovery. For such an algorithm, \cite{li2018rapid} derived near-optimal noise-bounds for a Gaussian noise model. However, as in this paper, we focus on the scenario of adversarial noise (instead of random noise) the resulting guarantees are not comparable to ours below.

\subsection{Matrix completion}
The matrix completion problem of reconstructing a low-rank matrix $X_0 \in \mathbb{R}^{n_1 \times n_2}$ (we assume that w.l.o.g.\ $ n_1 \ge n_2 $) from only a part of its entries arises in many different applications such as in collaborative filtering \cite{collabfiltering} and multiclass learning \cite{argyriou2008convex}. For this reason one could observe a flurry of work on this problem in the last decade, and we will only be able to give a very selective overview of this topic. The precise sampling model that we consider is that $m$ entries of $X_0$ are sampled uniformly at random with replacement. Denoting by $e_i$ the standard coordinate vectors in $ \mathbb{R}^{n_1} $ and $ \mathbb{R}^{n_2}$, respectively, the corresponding measurement operator $ \mathcal{A}: \mathbb{R}^{n_1 \times n_2}  \rightarrow \mathbb{R}^m$ can be written as
\begin{equation}\label{equ:MCopdefinition}
\mathcal{A} \left(X\right) \left(i\right):= \sqrt{ \frac{n_1 n_2}{m}} \innerproduct{X,e_{a_i}e^*_{b_i}}_F,
\end{equation}
where $\left(a_i,b_i\right) \in \left[n_1\right] \times \left[n_2\right]$ is chosen uniformly at random for each $ i \in \left[m\right] $ (and independently from all other measurements). The scaling factor $ \sqrt{ \frac{n_1 n_2}{m} } $ in the definition of the measurement operator $\mathcal{A} $ is chosen to ensure that $ \mathbb{E} \left[ \Vert \mathcal{A} \left( X \right) \Vert^2  \right] = \|X\|_F^2 $. (Some other papers on matrix completion choose a different scaling. We have chosen this normalization because in this way the results for the matrix completion problem can be better compared to those for the blind deconvolution scenario.) Alternative sampling models analyzed in other works include sampling a subset $ \Omega$ uniformly from $\left[n_1\right] \times \left[n_2\right] $ (i.e., without replacement, see, e.g., \cite{candes2010power}), or sampling using random selectors.\\ 

Again we aim to recover $X_0$ from noisy observations $ y = \mathcal{A} \left(X_0\right) + e $, with a noise vector $ e\in \mathbb{R}^m $ that satisfies $\|e\|\leq \tau$ via the SDP 
\begin{equation}\label{SDP_MC}
\begin{split}
\text{minimize } \quad  &\Vert X \Vert_{\ast} \\
\text{subject to} \quad & \Vert  \mathcal{A} \left(X\right) -y  \Vert \le \tau. 
\end{split}
\end{equation}
For matrix completion, this approach has first been studied in \cite{candes2009exact}.\\

It is well known that similarly to the blind deconvolution problem, some incoherence assumptions are necessary to allow for successful recovery. Indeed, suppose that $X_0 = e_1 e^*_1 $. Then, if $m\ll n_1 n_2 $ with high probability it holds that $ \mathcal{A} \left(X_0\right) =0 $ and one cannot hope to recover $X_0$. To avoid such special cases, one needs to ensure that the mass of the Frobenius norm of $X_0$ is spread out over all entries rather evenly. If $U\Sigma V^T$ is the singular value decomposition of  the rank-$r$ matrix $X_0$ (with $\Sigma \in \mathbb{R}^{r \times r}$), then this property is captured by the following coherence parameters \cite{gross2011recovering} 
\begin{align*}
\mu \left( U \right)&:= \sqrt{\frac{n_1}{r}} \underset{i \in \left[n_1\right]}{\max} \Vert U^* e_i \Vert \\
\mu \left( V \right)&:= \sqrt{\frac{n_2}{r}} \underset{i \in \left[n_2\right]}{\max} \Vert V^* e_i \Vert. 
\end{align*}
For these coherence parameters, a series of works \cite{candes2009exact,candes2010power,gross2011recovering,recht2011simpler,chen_coherenceoptimal,ding2018leave} lead to the following recovery guarantee for the noiseless scenario.
\begin{theorem}[\cite{ding2018leave}]\label{theorem:grossMC}
Consider measurements of the form $ y = \mathcal{A} \left(X_0\right)$, where $X_0\in \mathbb{R}^{n_1 \times n_2}$  is a rank-$r$ matrix with singular value decomposition $ X_0 = U\Sigma V^T $ and $\mathcal{A}$ is given by (\ref{equ:MCopdefinition}). Assume  that
	\begin{equation*}
	m \ge C  \max \left\{ \mu^2 \left(U\right); \mu^2 \left(V\right)   \right\}  r n_1 \log \left(n_1\right) \log \left[   r  \left( \mu \left(U\right) + \mu \left(V\right) \right) \right].
	\end{equation*}
	Then with probability at least $1 - \mathcal{O} \left(  n^{-1}_1 \right) $ the matrix $X_0$ is the unique minimizer of the SDP (\ref{SDP_MC}) with $\tau=0$.
\end{theorem}
As for blind deconvolution, this result has been shown using an approximate dual certificate. In \cite{candes2010matrix} this result has been generalized to the case of adversarial noise, showing that with high probability the minimizer $\hat{X}$ of (\ref{SDP_MC}) satisfies
\begin{equation}\label{ineq:errorcandesplan}
\Vert \hat{X} - X_0 \Vert_F \lesssim  \tau \sqrt{n_2} ,
\end{equation}
whenever $ m \gtrsim n_1 \text{polylog}\ n_1 $.
As in the blind deconvolution framework, this error bound differs from the case of full Gaussian measurements as discussed, for example, in \cite{chandrasekaran2012convex}, and also from oracle estimates \cite[Section III.B]{candes2011robust} by a dimensional scaling factor, which will be addressed in this paper. 

Also random noise models for matrix completion have been studied in a number of works. In particular, we would like to mention \cite{koltchinskii2011nuclear,negahban2012restricted}, which derive near-optimal rates (both in sample size and estimation error)  for matrix completion under subexponential noise with a slightly different nuclear-norm penalized estimator than the one we consider as long as the noise-level is not too small. Similar bounds have also been obtained in \cite{klopp} using an estimator, which is closer to the one in this work. 

 Apart from convex methods also many nonconvex algorithms have been proposed and analysed, for example a number of variants of gradient descent (see, e.g., \cite{optspace1,Jain_alternatingminimization,hardt2014MC,sun2016guaranteed,ge2016matrix,irlsfornasier,irlskuemmerle,ma2017implicit}).  Arguably the strongest result for matrix completion under adversarial noise has been shown in \cite{optspace1,optspace2}. These works propose a non-convex algorithm based on Riemannian optimization and show that if the number of measurements is larger than $ r^2 n_1 \text{polylog} \left( n_1 \right) $ the true matrix can be reconstructed up to an estimation error superior to the one in \cite{candes2010matrix}. Namely for $\kappa$ denoting the condition number of the matrix $X_0$ they show that the output $\hat{X}$ of their algorithm satisfies (in our notation)
\begin{equation}\label{ineq:boundmontanari}
\Vert \hat{X} -X_0 \Vert \lesssim \kappa^2 \sqrt{rm} \Vert e \Vert_{\infty},
\end{equation}
provided the noise level is below a certain, small threshold that scales with the smallest singular value of $X_0$. For error vectors $e$ that are spread out evenly and matrices that are well conditioned, one has that $\sqrt{m}\|e\|_\infty \approx \|e\|_2 $, so this bound is superior to (\ref{ineq:errorcandesplan}) in the sense that the scaling factors that appear only scale with the rank $r$ and not the dimension. It should be noted though that in contrast to nuclear norm minimization the underlying algorithm requires precise knowledge of the true rank of the matrix to be recovered.\\

Just before completion of this manuscript, Chen et al.\ \cite{chenconvex} bridged convex and nonconvex approaches, using nonconvex methods to analyze a convex recovery scheme. Their results provide near optimal recovery guarantees for the matrix completion problem via nuclear norm minimization under a subgaussian random noise model for a much larger range of admissible noise levels than the aforementioned works. More precisely, the proof is based on the observation that in their scenario the minimizer of the convex problem is very close to an approximate critical point of a non-convex gradient based method. This allows them to transfer existing stability results \cite{ma2017implicit} for non-convex optimization to the convex problem. However, the required sample complexity scales suboptimally in the rank $r$ of the matrix and similarly to (\ref{ineq:boundmontanari}), the error bound depends on the condition number $\kappa$.

\subsection{Descent cone analysis}\label{section:descentcone}
In recent years a number of works have studied low-rank matrix recovery and compressed sensing via a descent cone analysis. This approach has been pioneered for $\ell_1$-norm minimization in \cite{rudelson08} and for more general (atomic)  norms in \cite{chandrasekaran2012convex}. Here the descent cone of a norm at a point $X_0 \in \mathbb{C}^{K \times N} $ is the set of all possible directions $Z \in \mathbb{C}^{K \times N} $ such that the norm does not increase. For the nuclear norm, this leads to the following definition.
\begin{definition}
	For any matrix $X_0 \in \mathbb{C}^{K \times N}$ define its descent cone $ \dcone{X_0} $ by
	\begin{equation*}
	 \dcone{X_0} := \left\{  Z \in \mathbb{C}^{K \times N} : \ \Vert X_0+ \varepsilon Z \Vert_{\ast} \le \Vert X_0 \Vert_{\ast} \ \text{for some } \varepsilon>0  \right\}.
	\end{equation*}
\end{definition}
To understand its relevance for recovery guarantees assume for a moment that we are  in the noiseless scenario, i.e., $\tau =0$ and $ e = 0$. Then the matrix $X_0 \in \mathbb{C}^{K \times N}$ is the unique minimizer the semidefinite program (\ref{opt:SDP}), if and only if the null space of $ \mathcal{A}$ does not intersect the descent cone $  \dcone{X_0} $. In the case of noise, the constraint $\Vert y-\mathcal{A} \left(X_0\right) \Vert \le \tau $ in the SDPs (\ref{opt:SDP}) and (\ref{SDP_MC}) defines a region around $X_0 + \text{ker} \mathcal{A} $, i.e., the affine subspace consistent with the observed measurements in the noiseless scenario. The intersection of this region with the set of all signals that have a smaller nuclear norm than the ground truth $X_0$ is the set of feasible solutions that are preferred to $X_0$. The following quantity for a matrix $X_0$, which is often referred to as minimum conic singular value, quantifies the size of this intersection
\begin{equation*}
\lambda_{\min} \left( \mathcal{A},  \dcone{X_0}  \right) := \underset{Z \in  \dcone{X_0}  \setminus \left\{ 0 \right\}}{\inf} \frac{\Vert \mathcal{A} \left(Z\right) \Vert}{\Vert Z \Vert_F}.
\end{equation*}
If $\lambda_{\min} \left( \mathcal{A} ,   \dcone{X_0}  \right) $ becomes larger, this intersection becomes smaller, which translates into stronger recovery guarantees. The following theorem confirms this intuition.

\begin{theorem}\label{thm:chandrasekaran}\cite[Proposition 2.2]{chandrasekaran2012convex}
Let $\mathcal A: \mathbb{C}^{n_1\times n_2}\rightarrow \C^m$ be a linear operator and assume that $y= \mathcal{A} \left(X_0\right) + e $ with $ \Vert e \Vert \le \tau $. Then any minimizer $ \hat{X} $ of the SDP (\ref{opt:SDP}) satisfies
\begin{equation*}
\Vert \hat{X} - X_0 \Vert_F \le \frac{2 \tau}{\lambda_{\min} \left(\mathcal{A},  \dcone{X_0} \right) }.
\end{equation*}
\end{theorem}
When measurement matrices of the operator $ \mathcal{A} $ are full Gaussian matrices (in contrast to rank-$1$ measurements as in this paper) and $ \mathcal{A}$ is normalized such that $ \mathbb{E} \left[ \mathcal{A}^* \mathcal{A} \right] = \Id $, for an arbitrary low-rank matrix $X_0$ one has with high probability that $ \lambda_{\min} \left(\mathcal{A},  \dcone{X_0}  \right) \asymp 1 $. Consequently, Theorem \ref{thm:chandrasekaran} yields an optimal estimation error even for adversarial noise. As we will show this is no longer the case for blind deconvolution and matrix completion.\\

The geometric analysis of linear inverse problems via the descent cone and the minimum conic singular value has lead to many new results and insights in compressed sensing and low-rank matrix recovery.  For convex programs the phase transition of the success rate could be precisely predicted \cite{amelunxen2014}. As the proofs are specific to full Gaussian measuements, they do not apply for a number of important structured and heavy-tailed measurement scenarios. Stronger results \cite{lecue2017, dirksen2018gap, kueng2017low,kabanava2016stable,kueng2016low} were subsequently obtained using Mendelson's small ball method \cite{koltchinskii2015bounding,mendelson2014learning}, a powerful tool for bounding a nonnegative empirical process from below, now often refereed to as Mendelson's small ball method.

\subsection{Notation}
For $n\in \mathbb{N}$ we will write $\left[n\right] $ to denote the set $ \left\{ 1; \ldots; n \right\} $. For any set $A$ we will denote its cardinality by $\vert A \vert  $. For a complex number $z$ we will denote its real part by $ \text{Re} \left(z\right) $ and its imaginary part by $\text{Im} \left(z\right) $. By $ \log \left(\cdot\right) $ we will denote the logarithm to the base $e$. By $ \mathbb{E}X$ we will denote the expectation of a random variable $X$ and by $  \mathbb{P} \left( A \right) $ we denote the probability of an event $A$. If $v \in \mathbb{C}^n $ we will denote its $\ell_2$-norm by $\Vert v \Vert $ and its Hermitian transpose by $v^*$. For $u,v \in \mathbb{C}^n $ the (Euclidean) inner product is defined by $\innerproduct{u,v} := u^* v $. Furthermore, for $Z\in \mathbb{C}^{n_1 \times n_2}$ its spectral norm is given by $ \Vert Z \Vert$, i.e., the dual norm of the nuclear norm $\Vert Z \Vert_{\ast}$. Moreover, the Frobenius norm of $Z$ is defined by $\Vert Z \Vert_F $ with corresponding inner product $\innerproduct{Z,W}_F:= \trace \left(Z^* W\right) $, where $ W \in \mathbb{C}^{n_1 \times n_2} $. When we study matrix completion, we will work with matrices $ Z \in \mathbb{R}^{n_1 \times n_2} $ and the previous quantities will be defined analogously.  Moreover, in that scenario we will use the notation $\Vert Z \Vert_{\ell_{\infty}} := \underset{\left(i,j\right) \in \left[n_1\right] \times \left[n_2\right] }{\max} \vert Z_{i,j} \vert $, where $\left\{Z_{i,j} \right\}^{n_1,n_2}_{i,j=1}$. 

\section{Our results}\label{section:mainresults}

\subsection{Instability of low-rank matrix recovery}

\subsubsection{Blind deconvolution}

Our first main result states that randomized blind deconvolution can be unstable under adversarial noise.
\begin{theorem}\label{thm:mainresult}
Let $K,N \in \mathbb{N} \setminus \left\{ 1 \right\} $. Assume that 
\begin{equation*}
C_1  K \le L \le \frac{KN}{36}. 
\end{equation*}
Then there exists a matrix $B\in \mathbb{C}^{L\times K} $ satisfying $B^* B = \Id_K $ and with $FB$ having rows of equal norm, i.e., $\mu^2_{\max} =1 $, such that for all $h_0 \in \mathbb{C}^K \setminus \left\{ 0 \right\} $ and $m_0 \in \mathbb{C}^N \setminus \left\{ 0 \right\} $ the following holds:\\
With probability at least $1 - \mathcal{O} \left(\exp \left( -\frac{K}{C_2 \mu_{h_0}^2  } \right) \right)$, where $\mu^2_{h_0}=\frac{ L}{\Vert h_0 \Vert^2} \underset{\ell \in \left[L\right]}{\max} \ \vert \innerproduct{b_{\ell}, h_0} \vert^2 $, there is $\tau_0 >0$ such that for all $\tau \le \tau_0$ there exists an adversarial noise vector $e \in \mathbb{C}^{L} $ with $\Vert e \Vert \le \tau $ that admits an alternative solution $ \tilde{X}$ with the following properties.
	\begin{itemize}
		\item $\tilde{X}$ is feasible, i.e., $ \Vert \mathcal{A} \left(\tilde{X}\right) - y \Vert  = \tau $ for $ y= \mathcal{A} \left( h_0 m^*_0 \right) + e $ the noisy measurement vector
		\item $ \tilde{X}$ is preferred to $X_0 = h_0 m^*_0 $ by the SDP (\ref{opt:SDP}), i.e., $\Vert \tilde{X} \Vert_{\ast} \le \Vert X_0 \Vert_{\ast} $, but
		\item $\tilde{X}$ is far from the true solution in Frobenius norm, i.e.,
		\begin{equation*}
		\Vert \tilde{X} - X_0 \Vert_F \ge \frac{ \tau}{C_3}  \sqrt{ \frac{ KN}{L}}.
		\end{equation*}
	\end{itemize}
	The constants $ C_1 $, $C_2$, and $C_3$ are universal.
\end{theorem}

\begin{remark}
The matrix $B$ in the above result exactly fits into the framework of Theorem \ref{thm:ahmed}. Indeed, one can check that for many interesting cases (including the case that $K$ divides $L$) it holds that $ \tilde{\mu}_{h_0} \lesssim \mu_{h_0} $. That is, the assumptions of Theorem \ref{thm:ahmed} cannot be enough to deduce stability.\\ 
We do not expect, however, that this kind of instability is observed for \textit{arbitrary} isometric embeddings $B \in \mathbb{C}^{L \times K}$. For example, let $B$ be a random embedding, which is chosen from the uniform distribution over the Stiefel manifold $ \mathbb{V}^{L}_K $, i.e., the manifold consisting of all matrices $\tilde{B} \in \mathbb{C}^{L \times K}$ such that $\tilde{B}^*  \tilde{B}  = \Id_K \in \mathbb{C}^{K \times K} $. In this case, we expect that a similar proof as in \cite{RauhutTerstiege,kueng2017low} applies and that the multiplicative dimensional factor does not appear in the error bound with high probability if one randomizes over $B$ and $C$ simultaneously. In particular, this implies the existence of an isometric embedding $B$ such that a result analogous to Theorem \ref{thm:mainresult} cannot hold.\\
An interesting open problem is whether the statement of Theorem \ref{thm:mainresult} still holds if $FB$ is a low-frequency discrete Fourier matrix, which is a common assumption in blind deconvolution. The corresponding $b_{\ell}$'s should lead to better conditioning as in our counterexample, but worse than in the case of a random $B$, as a number of $b_{\ell}$'s exhibit substantial correlation, but many are uncorrelated. In that sense, this scenario is in between the scenario of arbitrary $B$'s sketched above and the adversarial scenario in Theorem \ref{thm:mainresult}.
\end{remark}
To put our results in perspective note that for $L \asymp \left(K+N \right) \text{polylog} \left(K+N\right)$, which is the minimal number of measurements required for noiseless recovery, it holds that $ \sqrt{ \frac{KN}{L}} \asymp \sqrt{\frac{\min \left\{ K,N \right\}}{\text{polylog} \left(K+N\right) }} $. Up to logarithmic factors, this coincides with the rate predicted by (\ref{ineq:noisebound}), whenever $ K \asymp N $.\\ 


Theorem \ref{thm:mainresult} is a direct consequence of the following proposition, which we think is interesting in its own right.
\begin{proposition}\label{theorem:mainpropBD}
	Let $K,N \in \mathbb{N} \setminus \left\{ 1 \right\} $. Assume that 
	\begin{equation}\label{equ:condition23456}
	C_1  K \le L \le \frac{KN}{36}. 
	\end{equation}
	Then there exists $B\in \mathbb{C}^{L\times K} $ satisfying $B^* B = \Id_K $ and $\mu^2_{\max} =1$,  whose corresponding measurement operator $ \mathcal{A}$ satisfies the following.\\
	Let $h_0 \in \mathbb{C}^K \setminus \left\{0\right\} $, $ m_0 \in \mathbb{C}^N \setminus \left\{0\right\} $ and define $ \mu_{h_0}$ as in Theorem \ref{thm:mainresult}. Then with probability at least $1- \mathcal{O} \left( \exp\left( -C_2 K/\mu_{h_0}^2 \right) \right) $ it holds that
	\begin{equation}\label{ineq:minconvalueBD}
	\lambda_{\min} \left(\mathcal{A},  \dcone{h_0 m^*_0}  \right) \le C_3 \sqrt{\frac{L}{KN}}.
	\end{equation}
	Here $C_1$, $C_2$, and $C_3$ are absolute constants.
\end{proposition}
The proof of Proposition \ref{theorem:mainpropBD} will be provided in Section \ref{section:proofconic}.
Note that by definition of the minimum conic singular value $ \lambda_{\min} \left( \mathcal{A},  \dcone{h_0m^*_0} \right) $ Proposition \ref{theorem:mainpropBD} is equivalent to the statement that with high probability there is $Z \in  \dcone{h_0 m^*_0} \setminus \left\{ 0 \right\} $ such that 
\begin{equation}\label{ineq:internyy}
\Vert \mathcal{A} \left(Z\right) \Vert  \lesssim  \sqrt{\frac{L}{KN}}  \Vert Z \Vert_F.
\end{equation}
Our construction of such $Z \in  \dcone{h_0 m^*_0} $ relies on the observation that with high probability there is a rank-one matrix $W\in \mathbb{C}^{K \times N}$ in the null-space of $ \mathcal{A} $ which is relatively close to the descent cone (with respect to the $\Vert \cdot \Vert_F$-distance). Perturbing $W$ by $-\beta h_0 m^*_0 $ for a suitable $\beta$ one can then obtain a matrix $Z\in  \dcone{h_0 m^*_0}  $, which fulfills (\ref{ineq:internyy}).\\
The existence of such a matrix $W\in \text{ker} \mathcal{A}$ also reveals a fact about the geometry of the problem, which we find somewhat surprising: while the null space of $\mathcal{A}$ does not intersect the descent cone (otherwise exact recovery would not be possible), the angle between those objects is very small. This is very different from the behavior for measurement matrices $ \mathcal{A}$ with i.i.d.~Gaussian entries (instead of $ b_{\ell} c^*_{\ell} $).

\begin{remark}\label{remark:actualminimizer}
While $\tilde{X}$ is preferred to the true solution by the SDP (\ref{opt:SDP})  $\tilde{X}$ is typically not a minimizer of (\ref{opt:SDP}). To see this, assume that without noise exact recovery is possible, which is the case with high probability by Theorem \ref{thm:ahmed}. Then consider $\tilde{X} = X_0 + tZ $ for $Z \in  \dcone{h_0 m^*_0}  $ of the form $Z=W - \beta h_0 m^*_0 $ with $W\in \text{ker} \mathcal{A}$ and $\beta >0 $ such that $ \frac{\Vert \mathcal{A} \left(Z\right) \Vert}{\Vert Z \Vert_F} \lesssim \sqrt{\frac{L}{KN}} $, as in the proof of Proposition \ref{theorem:mainpropBD}.
As $ W \notin  \dcone{h_0 m^*_0}  $ (otherwise exact recovery would not be possible) it follows that for $ t>0 $
\begin{align*}
	\Vert \tilde{X} \Vert_{\ast} &= \Vert X_0 + tZ \Vert_{\ast}\\
	&= \Vert \left(1-t\beta\right) X_0 + tW \Vert_{\ast}\\
	&>  \Vert \left(1-t\beta\right) X_0  \Vert_{\ast}
\end{align*}
where the last line is due to $  \dcone{X_0}   =  \dcone{ \left(1-t\beta\right) X_0}  $. 

On the other hand, we also have that $ \mathcal{A} \left( \hat{X} \right) = \mathcal{A}  \left(  \left(1-t\beta\right) X_0 \right) $ due to $\mathcal{A} \left(W\right) =0 $ and, hence, $\left(1-t\beta\right) X_0 $ is admissible whenever $ \tilde{X} $ is admissible. Consequently, the SDP (\ref{opt:SDP}) will always prefer $\left(1-\beta t\right) h_0 m^*_0 $ to $\tilde{X} $ and $ \tilde{X}$ will never be a minimizer. It remains an open problem what one can say about the minimizer $\hat{X}$ of (\ref{opt:SDP}), see also Section \ref{section:outlook}.\\
Even if the minimizer of (\ref{opt:SDP}) $\hat{X}$ is closer to the ground truth (in $\Vert \cdot \Vert_F$-distance) than $\tilde{X}$, however, the nuclear norms of $X$ and $\tilde{X}$ will be very close, which can easily lead to numerical instabilities. 
\end{remark}

\subsubsection{Matrix completion}
Our second main result states that for arbitrary incoherent low-rank matrices, matrix completion is unstable with high probability. Note that in contrast to Theorem \ref{thm:mainresult} which is based on a specific choice of parameters the following result holds for an arbitrary incoherent matrix $X_0$.
\begin{theorem}\label{thm:mainresultMC}
	Let $n_1 \ge n_2 $ and let $\mathcal{A}: \mathbb{R}^{n_1 \times n_2} \rightarrow \mathbb{R}^m  $ be defined as in (\ref{equ:MCopdefinition}). Assume that $ X_0 \in \mathbb{R}^{n_1 \times n_2} \setminus \left\{ 0 \right\} $ is a rank $r$ matrix with singular value decomposition $X_0 = U\Sigma V^* $. Moreover, assume that
	\begin{equation*}
	C_1 rn_1 \mu^2 \left(V\right) \log (2r) \le  m \le \frac{n_1 n_2}{32}.
	\end{equation*}
	Then with probability at least $1 - \mathcal{O} \left(\exp \left( -\frac{m}{C_2r \mu^2 \left( U \right) \mu^2 \left( V \right)  } \right) \right)$ there is $\tau_0 > 0 $ such that for all $\tau \le \tau_0 $ there exists an adversarial noise vector $e \in \mathbb{R}^{m} $ with $\Vert e \Vert \le \tau $ that admits an alternative solution $\tilde{X} \in \mathbb{R}^{n_1 \times n_2} $ with the following properties.
	\begin{itemize}
		\item $\tilde{X}$ is feasible, i.e., $\Big\Vert \mathcal{A} \left( \tilde{X} \right) - y \Big\Vert=  \tau $ for $y= \mathcal{A} \left(X_0\right)+e $ the noisy measurement vector
		\item $\tilde{X}$ is preferred to $X_0 $ by the SDP (\ref{SDP_MC}), i.e.,    $\Vert \tilde{X} \Vert_{\ast} \le  \Vert X_0 \Vert_{\ast} $ , but
		\item $\tilde X$ is far from the true solution in Frobenius norm, i.e.,
		\begin{equation*}
		\Vert \tilde{X} - X_0 \Vert_F \ge \frac{ \tau}{C_3}  \sqrt{ \frac{ r  n_1n_2}{m}}.
		\end{equation*}
	\end{itemize}
	Here the constants $ C_1 $, $C_2$, and $C_3$ are universal.
\end{theorem}
Again, to put our results in perspective note that for $m \asymp n_1 \text{polylog} \left( n_1 \right)$, which is the minimal number of measurements required for noiseless recovery, it holds that $ \sqrt{ \frac{rn_1 n_2}{m}} \asymp \sqrt{\frac{n_2}{\text{polylog} \left(n_1\right) }} $. Up to logarithmic factors, this coincides with the rate predicted by (\ref{ineq:errorcandesplan}).\\ 


Theorem \ref{thm:mainresultMC} is a direct consequence of the following proposition, which, in our opinion, is of independent interest, as it provides a negative answer to a question by Tropp \cite[Section 5.4]{tropp2015convex}.
\begin{proposition}\label{theorem:mainprop}
	Let $ X_0 \in \mathbb{R}^{n_1 \times n_2} \setminus \left\{ 0 \right\} $ be a rank-$r$ matrix with corresponding singular value decomposition $X_0 = U\Sigma V^* $. Moreover, assume that
	\begin{equation}\label{ineq:assumptionbla2}
	C_1 rn_1 \mu^2 \left(V\right) \log (2r) \le  m \le \frac{n_1 n_2}{32}.
	\end{equation}
	Then with probability at least $1 - \mathcal{O} \left(\exp \left( -\frac{m}{C_2 r \mu^2 \left( U \right) \mu^2 \left( V \right)  } \right) \right) $ it holds that
	\begin{equation}\label{ineq:minconvalueMC}
	\lambda_{\min} \left(\mathcal{A},  \dcone{X_0} \right) \le  C_3   \sqrt{ \frac{ m}{n_1 n_2 r}}.
	\end{equation}
	The constants $ C_1 $, $C_2$, and $C_3$ are universal.
\end{proposition}
Proposition \ref{theorem:mainprop} corresponds to Proposition \ref{theorem:mainpropBD} for blind deconvolution and will be proved analogously. We will again show that with high probability there is $W \in \mathbb{R}^{n_1 \times n_2}$ such that $\mathcal{A} \left(W\right) = 0$ and $W$ is relatively close to the descent cone of $ X_0$ in $\Vert \cdot \Vert_F$-distance. Setting $Z:= W-\beta UV^*$ for a suitable $\beta>0$ yields an element of $ \dcone{X_0}$ with
\begin{equation*}
 \frac{\Vert \mathcal{A} \left(Z\right) \Vert }{\Vert Z \Vert_F}  \le  C_3   \sqrt{ \frac{ m}{n_1 n_2 r}}.
\end{equation*}


\subsection{Stable recovery}


\begin{figure}
	\centering
	\begin{tikzpicture}[scale=0.625]
	\draw [draw,  pattern=north west lines, pattern color=orange] (-7.5cm,4.4cm)--(12.5cm,2.2cm)--(12.5cm,5.8cm)--(-7.5cm,7.8cm);
	\path[draw, line width=1pt, fill=blue, opacity=0.1]  (-9cm,2.0cm)--(2.5cm,5cm)--(14cm,2.0cm);
	\path[draw, line width=1pt, fill=green, opacity=0.3]  (-2.2cm,2.0cm)--(2.5cm,5cm)--(7.2cm,2.0cm);
	\path[draw, line width=2.5pt] (-7.5cm,6.1cm)--(12.5cm,3.9cm);
	\draw[red, pattern=north west lines, pattern color=red, line width=2.5pt] (-1.3cm,2.0cm)--(0.7cm,4.0cm) .. controls  (2.5cm,5.3cm) .. (4.3cm,4.0cm)--(6.3cm,2.0cm);
	\node (A) at (2.5cm,5.8cm)  { $X_0$};
	\node(B) at (-7.5cm,6.9cm) {$X_0+ \text{ker} \mathcal{A}$};
	\end{tikzpicture}


\caption{Geometric illustration of our approach: Close to $X_0$ the descent set (indicated by the red line) is near-tangential to the kernel of the measurement operator $\mathcal{A}$, so the descent cone (light blue) is rather wide. By restricting to noise levels above a certain threshold we only need to cover the descent set at some distance, which is achieved by a much smaller cone (green). Note that below the noise level (orange strip), the green cone does not contain the full descent set.}
\end{figure}
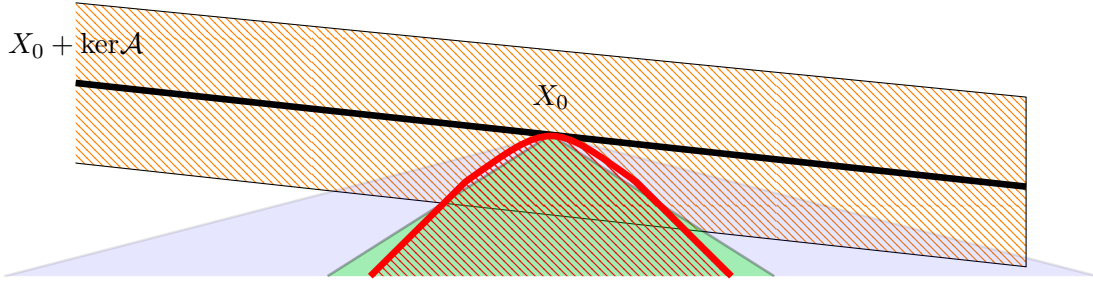

A geometric interpretation of Theorems \ref{thm:mainresult} and \ref{thm:mainresultMC} is that the nuclear norm ball is near-tangential to both the kernels of matrix completion and randomized blind deconvolution. Given that tangent spaces only provide local approximation, these results leave open, what happens in some distance, i.e., for larger noise levels -- this will depend on the curvature of the nuclear norm ball.

Our third main result concerns exactly this problem for the randomized blind deconvolution setup. As it turns out, the descent directions $ Z \in  \dcone{h_0 m^*_0}  $ with $\Vert \mathcal{A} \left(Z\right) \Vert / \Vert Z \Vert_F $ very small correspond to directions of significant curvature. That is, only a very short segment in this direction will have smaller nuclear norm than $h_0 m^*_0$, and the corresponding alternative solutions all correspond to very small $e$. For noise levels $\tau$ large enough, in contrast, these directions can be excluded and one can obtain near-optimal error bounds. In order to precisely formulate this observation, we denote the set of $\mu$-incoherent vectors $h\in \mathbb{C}^K$ with respect to $B \in \mathbb{C}^{L \times K}$ for $\mu \ge 1 $ by
\begin{equation*}
\mathcal{H_{\mu}}:= \left\{ h_0 \in \mathbb{C}^K: \ \ \sqrt{L} \vert \innerproduct{ b_{\ell},h_0 } \vert \le \mu \Vert h_0 \Vert \text{ for all } \ell \in \left[L\right]   \right\}.
\end{equation*}
With this notation, our result reads as follows.
\begin{theorem}\label{thm:stabilityBD}
	Let $ \alpha>0 $, $ \mu \ge 1$, and $ B \in \mathbb{C}^{L\times K} $ such that $B^* B= \Id $. Assume that
	\begin{equation*}
	L \ge C_1  \frac{\mu^2}{\alpha^2}  \left(K+N\right) \ \log^2 L.
	\end{equation*}
	Then with probability at least $ 1 - \mathcal{O} \left( \exp \left( - \frac{L \alpha^{4/3} }{C_2 \log^{4/3} \left(eL\right) \mu^{4/3 } }  \right) \right) $ the following statement holds for all $ h_0 \in \mathcal{H}_{\mu} \setminus \left\{ 0 \right\} $, all $m_0 \in \mathbb{C}^N \setminus \left\{ 0 \right\} $, all $ \tau > 0 $, and all $ e \in \mathbb{C}^L $ with $ \Vert e \Vert \le \tau$ :\\
	Any minimizer $\hat{X}$ of (\ref{opt:SDP}) satisfies
	\begin{equation*}
	\Vert \hat{X} - h_0 m^*_0 \Vert_F \le \frac{ C_3 \mu^{2/3}  \log^{2/3} L}{\alpha^{2/3}} \max \left\{  \tau  ; \alpha \Vert h_0 m^*_0 \Vert_F  \right\}.
	\end{equation*}
	Here $C_1$, $C_2$, and $C_3$ are absolute constants.
\end{theorem}
In words, this theorem establishes linear scaling in the noise level $\tau$ with only a logarithmic dimensional factor for $\tau \ge \alpha \Vert h_0 m^*_0 \Vert_F$, in contrast to the polynomial factor required for small noise levels as a consequence of Theorem \ref{thm:mainresult}. Here the value of $\alpha$ can be chosen arbitrarily small, at the expense of an increased number of measurements. For example when one is interested in noise levels $\tau = \epsilon \mu^{-2} \log^{-2} L$ for some $\epsilon>\epsilon_0$ (this is the largest order to expect meaningful error bounds despite the additional logarithmic factors) one should choose $\alpha \asymp \epsilon_0 \mu^{-2} \log^{-2} L$, and near-linear error bounds will be guaranteed for a sample complexity of 
	\begin{equation*}
L \ge C_1  \frac{\mu^6}{\epsilon_0^2}  \left(K+N\right) \ \log^6 L.
\end{equation*}
\begin{remark}
A similar approach to the proof of Theorem \ref{thm:stabilityBD} also yields a corresponding result for rank-one matrix completion. Arguably, however, matrix completion is mainly of  interest for ground truth matrices of rank higher than one, so we decided to omit the proof details.
\end{remark}

\section{Upper bounds for the minimum conic singular values}\label{section:proofconic}
\subsection{Characterization of the descent cone of the nuclear norm}

The goal of this section is to prove Proposition \ref{theorem:mainpropBD} and Proposition \ref{theorem:mainprop}, from which we will then be able to deduce Theorem \ref{thm:mainresult} and Theorem \ref{thm:mainresultMC}. For that we first discuss a characterization of the descent cone $ \dcone{X}  $.
In order to state this characterization, Lemma \ref{lemma:descentconechar}, we need to introduce some additional notation. Let $ X \in \mathbb{C}^{n_1 \times n_2} $ be a matrix of rank $r$. We will denote its corresponding singular value decomposition by $ X= U \Sigma V^* $, where $ \Sigma \in \mathbb{R}^{r \times r}  $ is a diagonal matrix with nonnegative entries and $U \in \mathbb{C}^{ n_1 \times r }$ and $V \in \mathbb{C}^{n_2 \times r}$ are unitary matrices, i.e., $  U^*U  = V^* V = \Id_r$. This allows us to define the tangent space of the manifold of rank-$r$ matrices at the point $X$ by
\begin{equation}\label{equ:deftangentspace}
T_X := \left\{  UA^* + B V^*    : \ A \in \mathbb{C}^{n_2 \times r}, B \in \mathbb{C}^{n_1 \times r}     \right\}.
\end{equation}
By $ \mathcal{P}_{T_X} $ we will denote the orthogonal projection onto $ T_X $, by $ \mathcal{P}_{T^{\perp}_X} =   \Id - \mathcal{P}_{T_X} $ the projection onto its orthogonal complement.
\begin{lemma}\label{lemma:descentconechar}
		Let $ X \in \mathbb{C}^{n_1 \times n_2} \backslash \left\{ 0 \right\}$ be a matrix of rank $r$ with corresponding singular value decomposition $ X= U \Sigma V^* $. Then
	\begin{equation*}
	\overline{   \dcone{X}  } =  \left\{ Z \in \mathbb{C}^{n_1 \times n_2} :   - \text{Re} \left( \innerproduct{UV^*, Z}_F   \right)    \ge \Vert \mathcal{P}_{T^{\perp}_X} \left( Z \right)   \Vert_{\ast}     \right\},
	\end{equation*} 
	where $\overline{   \dcone{X}  } $ denotes the topological closure of $   \dcone{X}   $.
\end{lemma}
\begin{remark}
Lemma \ref{lemma:descentconechar} is similar to well-known results in convex optimization and may be known to the community. As we could not find it in the literature in this form, we decided to include a proof for completeness.
\end{remark}
The proof of Lemma \ref{lemma:descentconechar} relies on the duality between the descent cone and the subdifferential of a convex function. In the following we will denote by $ \partial \Vert \cdot \Vert_{\ast} \left(X\right) $ the subdifferential of the nuclear norm at the point $X \in \mathbb{C}^{n_1 \times n_2} $. We will use that a characterization of $ \partial \Vert \cdot \Vert_{\ast} $ is well-known \cite{watson_subdifferential}. Namely, for all $ X \in \mathbb{C}^{n_1 \times n_2} $ with corresponding singular value decomposition $ X=U\Sigma V^* $ it holds that
\begin{equation}\label{equ:charsubdifferential}
\partial \Vert \cdot \Vert_{\ast} \left( X \right) = \left\{  W \in \mathbb{C}^{n_1 \times n_2} :  \ \mathcal{P}_{T_X} W = UV^*, \   \Vert \mathcal{P}_{T^{\perp}_X} W \Vert\le 1         \right\}.
\end{equation}
\begin{proof}
	Recall that for a set of matrices $ \mathcal{V} \subset \mathbb{C}^{n_1 \times n_2} $ its polar cone $ \mathcal{V}^{\circ} $ is defined by
	\begin{equation*}
	\mathcal{V}^{\circ}  := \left\{ Z \in \mathbb{C}^{n_1 \times n_2} : \text{Re} \left( \innerproduct{ W,Z }_F \right)\le 0  \text{ for all } W \in \mathcal{V}   \right\}.
	\end{equation*}
	For all $ X \in \mathbb{C}^{n_1 \times n_2} \backslash \left\{ 0 \right\} $ we have the following polarity relation between the descent cone and the subdifferential 
	\begin{equation*}
	 \dcone{X}^{\circ} = \overline{ \left\{ \lambda W: \ \lambda \ge 0, \ W \in \partial \Vert \cdot \Vert_{\ast} \left( X \right)    \right\}}.
	\end{equation*}
	For sets and functions defined in $ \mathbb{R}^n$ with the usual Euclidean inner product, this is \cite[Theorem 23.7]{rockafellar1970}. The complex case directly follows, as $\mathbb{C}^{n_1 \times n_2} $ with the inner product $ \text{Re} \left( \innerproduct{\cdot, \cdot}_F \right) $ can be identified with an $2n_1 n_2 $-dimensional real-valued vector space with standard Euclidean inner product.\\
	
	\noindent It follows from the bipolar theorem (see, e.g., \cite[p. 53]{boyd_optimization}) that
	\begin{equation*}\label{equ:dualitysubdifferential}
	\overline{   \dcone{X} } =  \left(  \partial \Vert \cdot \Vert_{\ast} \left( X \right)  \right)^{\circ}.
	\end{equation*}
	Hence, in order to complete the proof it is sufficient to show that
	\begin{equation}\label{equ:necessarycondition}
	\left\{ Z \in \mathbb{C}^{n_1 \times n_2} :   - \text{Re} \left( \innerproduct{UV^*, Z}_F   \right)    \ge \Vert \mathcal{P}_{T^{\perp}_X} \left( Z \right)   \Vert_{\ast}     \right\} =\left(  \partial \Vert \cdot \Vert_{\ast} \left( X \right)  \right)^{\circ} = \overline{  \text{cone} \left( \partial \Vert \cdot \Vert_{\ast} \left( X \right)  \right)  }  .
	\end{equation}
	First, suppose that $Z \in \mathbb{C}^{n_1 \times n_2} $ satisfies $- \text{Re} \left( \innerproduct{UV^*, Z}_F   \right)    \ge \Vert \mathcal{P}_{T^{\perp}_X} \left( Z \right)   \Vert_{\ast}$. We have to show that $ \text{Re} \left( \innerproduct{W,Z}_F \right) \le 0  $ for all $ W \in \partial \Vert \cdot \Vert_{\ast} \left( X \right)  $. Indeed, 
	\begin{align*}
	\text{Re} \left( \innerproduct{W,Z}_F \right) &= \text{Re} \left( \innerproduct{\mathcal{P}_{T_X} W ,Z}_F \right)+ \text{Re} \left( \innerproduct{\mathcal{P}_{T^{\perp}_X} W ,Z}_F \right) \\
	&= \text{Re} \left( \innerproduct{UV^*,Z}_F \right) + \text{Re} \left( \innerproduct{\mathcal{P}_{T^{\perp}_X} W ,  \mathcal{P}_{T^{\perp}_X} Z}_F \right) \\
	&\le \text{Re} \left( \innerproduct{UV^*,Z}_F \right)  +  \Vert \mathcal{P}_{T^{\perp}_X} W \Vert  \Vert  \mathcal{P}_{T^{\perp}_X} Z \Vert_{\ast}  \\
	&  \le \text{Re} \left( \innerproduct{UV^*,Z}_F \right)  +  \Vert  \mathcal{P}_{T^{\perp}_X} Z \Vert_{\ast}  \\
	& \le 0.
	\end{align*}
	In the first inequality we have used that the spectral norm is the dual norm of the nuclear norm. The second inequality follows from $  \Vert \mathcal{P}_{T^{\perp}_X} W \Vert \le 1 $. Hence, we have shown that $ Z \in \left(  \partial \Vert \cdot \Vert_{\ast} \left( X \right)  \right)^{\circ} $. Next, let $Z\in   \left(  \partial \Vert \cdot \Vert_{\ast} \left( X \right)  \right)^{\circ}$ be arbitrary.  Choose $ \tilde{W} \in T^{\perp}_{X} $ such that $ \text{Re} \left(  \innerproduct{\tilde{W},Z}_F \right)  = \Vert \mathcal{P}_{T^{\perp}_X} \left( Z \right)   \Vert_{\ast}  $ and $ \Vert \tilde{W} \Vert \le 1 $. Then by (\ref{equ:charsubdifferential}) it follows that $ UV^* + \tilde{W} \in \partial \Vert \cdot \Vert_{\ast} \left( X \right)   $ and as $ Z \in \left( \partial \Vert \cdot \Vert_{\ast} \left( X \right) \right)^{\circ} $ we obtain that
	\begin{align*}
	0 &\ge \text{Re} \left(  \innerproduct{UV^* + \tilde{W}, Z}_F \right) \\
	& =  \text{Re} \left(  \innerproduct{UV^*, Z}_F \right)  + \Vert   \mathcal{P}_{T^{\perp}_X} \left( Z \right)  \Vert_{\ast}.
	\end{align*}
	This shows that  $-\text{Re} \left( \innerproduct{UV^*, Z}_F   \right)    \ge \Vert \mathcal{P}_{T^{\perp}_X} \left( Z \right)   \Vert_{\ast}  $. Hence, we have verified (\ref{equ:necessarycondition}), which completes the proof.
\end{proof}

\subsection{Upper bound for blind deconvolution}
The goal of this section is to prove Proposition \ref{theorem:mainpropBD}. For that we  need the following lemma, which is a consequence of the concentration of measure theorem for Lipschitz functions. (For a proof of the real-valued case see, e.g., \cite[Lemma 5.3.2]{vershynin2016high}. The complex-case can be shown analogously.)
\begin{lemma}\label{lemma:projectionsubspace}
	Let $P: \mathbb{C}^n \rightarrow \mathbb{C}^n$ be a random projection onto a $k$-dimensional subspace, which is uniformly distributed in the Grassmannian $Gr \left(k, \mathbb{C}^n \right) $. Fix $ z \in \mathbb{C}^n $. Then for all $\varepsilon > 0$ with probability at least $ 1- 2 e^{-\tilde{c} k\varepsilon^2} $ we have that
	\begin{equation*}
	\left( 1-\varepsilon \right)  \frac{k}{n} \Vert z \Vert^2 \le  \Vert Pz \Vert^2 \le \left( 1+\varepsilon \right)  \frac{k}{n}  \Vert z \Vert^2,
	\end{equation*}
	where $\tilde{c}>0$ is some absolute constant.
\end{lemma}
\begin{proof}[Proof of Proposition \ref{theorem:mainpropBD}]
	A core ingredient of the proof is to find a tight frame $B$ such that each of its frame vectors is orthogonal to all but a near-minimal number of other frame vectors. For such a $B$ we then choose a vector $h$ out of these frame vectors and use it to construct a matrix in the descent cone that is close to the kernel of the measurement map. For that, we exploit that by the choice of $h$, any rank-one matrix $h\tilde{m}^*$ will lead to a large part of zero measurements due to the orthogonality. Consequently, there are also many vectors $m_1$ and $m_2$ such that $hm_1^*$ and $hm_2^*$ lead to the same measurements, including some choices such that $hm_1^* - hm_2^* $ is not only in the kernel of the measurement map, but also close to the descent cone.
	
	When $K$ divides $L$ a suitable choice for $B$ consists of $\tfrac{L}{K}$ repetitions of a fixed orthogonal basis. When $K$ does not divide $L$, one can still start off in the same way, if one completes the matrix appropriately to obtain a unit norm tight frame without introducing too many pairs of non-orthogonal frame vectors. One way to achieve this is to find a tight frame that consists only of very sparse vectors, as this will also lead to many vanishing inner products. A natural candidate is hence a so-called spectral tetris frame \cite{spectraltetris1}, as it has been shown to be maximally sparse \cite{spectraltetris2}.  Indeed, our construction uses exactly this frame. 
	
	The Spectral Tetris algorithm is based on the observation that the rows of a matrix $G\in \mathbb{R}^{L\times K}$ form a tight frame if and only if the columns of $G$ are orthogonal and of equal norm. It is easy to see that for  unit norm tight frames, the column normalization must be $\sqrt{\frac{L}{K}} $. Spectral Tetris starts by greedily filling up the first column of $G$ by choosing the first $\lfloor \tfrac{L}{K}\rfloor$ rows of $G$ to be the first standard basis vector $e_1$. The next two frame vectors are chosen of the form $\alpha e_1 \pm \sqrt{1-\alpha^2}e_2$, where $\alpha$ is chosen to fulfil the norm constraint of the first column. Next, the second column is greedily filled, then the third, and so on. By construction, the resulting frame will only consist of $1$-sparse vectors and $2$-sparse vectors supported in neighboring entries. Consequently, the sets
	\begin{equation*}
	\mathcal{B}_i := \left\{ b_{\ell}: \innerproduct{b_{\ell}, e_{1+ 3 \left( i-1 \right)}} \ne 0   \right\},   1\le i \le \frac{K}{3} 
	\end{equation*}
	are disjoint.

	
	
	Moreover, note that it follows from the Spectral Tetris algorithm that  
	\begin{align}\label{ineq:lastineq}
	\frac{L}{K} \le   \vert \mathcal{B}_i \vert \le  3 \left(\frac{L}{K} + 2\right) \le 6 \frac{L}{K},
	\end{align}
	where in the last inequality we used that by assumption $ 2K \le L $.\\
	
	\noindent Without loss of generality we assume that $\Vert h_0 \Vert = \Vert m_0 \Vert =1 $ as rescaling does not change the descent cone $  \dcone{h_0 m^*_0}$. For the proof we will condition on two events. The first event states that
	\begin{equation}\label{ineq:event1}
	\Vert \mathcal{A} \left( h_0 m^*_0 \right) \Vert < 2 \Vert h_0 m^*_0 \Vert_F,
	\end{equation} 
	which by the Bernstein inequality (see, e.g., \cite{vershynin2016high}) is fulfilled with probability at least $1-\exp \left( -cL/\mu^2 \right) $, where $c>0$ is some numerical constant.  To formulate the second event, we define for all natural numbers $ 1 \le i \le K/3$ and $c_{j}$ as in Section~\ref{section:backgroundBD}
	\begin{equation*}
	D_{i}:= \text{span} \left\{ c_{i}: i \in \mathcal{B}_i \right\} \subset \mathbb{C}^N
	\end{equation*}
	and denote by $ m_i^{\parallel}$ the orthogonal projection of $m_0$ onto $D_i$ and by $ m_i^{\perp}$ the projection onto $D_i^{\perp}$, the orthogonal complement of $D_i$. Note that $D_i \subset \mathbb{C}^N $ is a random subspace of dimension $\vert \mathcal{B}_i \vert$, distributed uniformly over the Grassmannian $Gr \left( \vert \mathcal{B}_i \vert, \mathbb{C}^N  \right) $ due to the rotation invariance of $ \mathcal{CN} \left(0,1\right) $. Hence, as $\Vert m_0 \Vert =1 $ Lemma \ref{lemma:projectionsubspace} together with inequality \eqref{ineq:lastineq} yields that for fixed $ i \in \left[ \lfloor K/3 \rfloor  \right] $ with probability at least $ 1-2\exp \left(-  \frac{ \hat{c} L }{4K} \right)$ one has
	\begin{equation}\label{ineq:event2}
	\frac{\vert \mathcal{B}_i \vert}{2N}  \le  \big\Vert m_i^{\parallel} \big\Vert^2 \le \frac{3 \vert \mathcal{B}_i \vert}{2N}  \le   6 \frac{L}{KN},
	\end{equation}
	where $\hat{c}>0$ is some absolute constant. As the matrix $C$ is Gaussian, the different subspaces $D_i$'s and hence also the random vectors $\left\{m_i^{\parallel} \right\}^K_{i=1}$ are independent, so with probability at least 
	\begin{equation*}
	1- \left(2\exp \left(-  \frac{ \hat{c}L }{4K} \right) \right)^{\lfloor K/3 \rfloor} \ge 1-  \exp \left( K \log 2 -  \frac{ \hat{c}L }{12} \right)
	\end{equation*}
	there exists at least one $ k \in \left[ \lfloor K/3  \rfloor \right]  $ such that (\ref{ineq:event2}) holds (with $k=i$). Also note that
	\begin{equation*}
	1-  \exp \left( K \log 2 -  \frac{ \hat{c}L }{12} \right)  \ge 1-\exp \left( - \frac{\hat{c}}{24}L \right),
	\end{equation*}
	which for $C_1=   \frac{ 24 \log 2}{\hat{c}} $ follows from assumption (\ref{equ:condition23456}). 
	
	To summarize, we have shown that the two events  $ \mathcal{E}_1 := \left\{ 	\Vert \mathcal{A} \left( h_0 m^*_0 \right) \Vert \le 2 \Vert h_0 m^*_0 \Vert_F  \right\} $ and 
	\begin{equation*}
	\mathcal{E}_2 := \left\{ \exists i\in \left[K\right]: \ \frac{L}{2KN}  \le  \big\Vert m_i^{\parallel} \big\Vert^2 \le  6 \frac{L}{KN}  \right\} 
	\end{equation*}
	happen with probability at least $ 1-  \mathcal{O} \left( \exp \left(-C_2 L/\mu^2\right) \right) $, where $C_2>0$ is an appropriately chosen constant.\\
	\\ 
	
	Conditional on $ \mathcal{E}_1 $ and $ \mathcal{E}_2 $, we will construct $Z \in \mathbb{C}^{K \times N} $ (depending on the realization of the random matrix $C$) such that $Z \in \overline{   \dcone{h_0 m^*_0}   } \setminus \left\{ 0 \right\} $ and such that the inequality
	\begin{equation}\label{ineq:intern12345}
	\Vert \mathcal{A} \left( Z \right) \Vert < 12 \sqrt{ \frac{L}{KN} } \Vert Z \Vert_F
	\end{equation}
	is satisfied. Note that this will complete the proof. Indeed, by definition of the closure and the continuity of $\mathcal{A}$ this implies that there exists $\tilde{Z} \in  \dcone{X_0}  $ such that 
	\begin{equation*}
	\frac{ \big\Vert \mathcal{A} \left( \tilde{Z}\right) \big\Vert} { \Vert \tilde{Z} \Vert_F } \le 12 \sqrt{\frac{L}{KN}}, 
	\end{equation*}
	which by the definition of $\lambda_{\min} \left(\mathcal{A},  \dcone{h_0 m^*_0}  \right) $ implies that (\ref{ineq:minconvalueBD}) holds with constant $C_3=12$.\\
	
	\noindent
	To construct $Z$ satisfying \eqref{ineq:intern12345}, define 
	\begin{equation*}
	W:=- \frac{\innerproduct{ h_0,e_i }}{\Vert m_i^{\perp} \Vert \vert \innerproduct{e_i,h_0} \vert} e_i \left(m_i^{\perp}\right)^*,
	\end{equation*}
	where $ i\in \left[ \lfloor K/3 \rfloor  \right] $ is chosen to satisfy (\ref{ineq:event2}). It follows directly from the definition of $W$ that $\Vert W \Vert_F =1$. We observe that $\mathcal{A} \left(W\right) =0 $ as for each $i\in \left[ \lfloor K/3 \rfloor \right] $  and $ \ell \in \left[ L \right] $ we either have $  \innerproduct{e_i, b_{\ell}}  =0$, $\ell \notin \mathcal{B}_i$, or $ \innerproduct{m^{\perp}_{i}, c_{\ell} } =0 $, if $\ell \in \mathcal{B}_i$. Denote by $T=T_{X_0} $ the tangent space of the manifold of rank-one matrices at $X_0 = h_0 m^*_0$ as defined in (\ref{equ:deftangentspace}) and by $  \mathcal{P}_T $ and $  \mathcal{P}_{T^{\perp}} $ the corresponding orthogonal projections. It follows that
	\begin{equation}\label{ineq:internabcd}
	\begin{split}
	\Vert \mathcal{P}_{T^{\perp}} W \Vert_{\ast}&= \Big\Vert  \mathcal{P}_{T^{\perp}}  \left( e_i \left(\frac{ m_i^{\perp}}{\Vert m_i^{\perp} \Vert} \right)^* \right)  \Big\Vert_F \\
	&= \Big\Vert P_{h^{\perp}_0} e_i \Big\Vert \Big\Vert P_{m^{\perp}_0} \left(\frac{ m_i^{\perp}}{\Vert m_i^{\perp} \Vert} \right)  \Big\Vert  \\
	&= \sqrt{1- \vert \innerproduct{h_0,e_i} \vert^2} \sqrt{1 - \Big\vert \innerproduct{m_0, \frac{ m_i^{\perp}}{\Vert  m_i^{\perp} \Vert } }\Big\vert^2} \\
	&\le  \sqrt{1- \big\vert \innerproduct{m_0, \frac{ m_i^{\perp}}{\Vert  m_i^{\perp} \Vert }   }\big\vert^2} \\
	&= \sqrt{1- \Vert m_i^{\perp} \Vert^2}\\
	&=  \big\Vert m_i^{\parallel} \big\Vert   \\
	&\overset{\left( \ref{ineq:event2} \right)}{\le} \sqrt{ 6  \frac{L}{KN} }.
	\end{split}
	\end{equation}
	Thus we have shown that $W$, an element of the null space of $\mathcal{A} $, is close to the tangent space $T$. We will now show that for $ \beta = 3 \sqrt{ \frac{L}{KN} } $
	\begin{equation*}
	Z:= - \beta h_0 m^*_0 +W,
	\end{equation*}
	lies in the closure of the descent cone $ \overline{ \dcone{h_0 m^*_0}}$. 
	%
	For that, we observe that
	\begin{align*}
	-\text{Re} \left( \innerproduct{Z, h_0 m^*_0} \right) &= \beta -  \text{Re} \left( \innerproduct{W, h_0 m^*_0}_F  \right)\\
	&= \beta +  \text{Re} \left( \frac{\innerproduct{h_0,e_i }}{\Vert m_i^{\perp} \Vert \vert \innerproduct{h_0,e_i} \vert}  \innerproduct{ e_i \left( m_i^{\perp} \right)^* , h_0 m^*_0}_F \right)\\
	&= \beta +  \text{Re} \left( \frac{\innerproduct{ h_0,e_i }}{\Vert m_i^{\perp} \Vert \vert \innerproduct{h_0,e_i} \vert}  \innerproduct{e_i, h_0} \innerproduct{m_0, m_i^{\perp}}  \right)\\
	&= \beta + \vert \innerproduct{h_0,e_i} \vert  \Vert m_i^{\perp} \Vert  \\
	&\ge \beta\\
	&\overset{(\ref{ineq:internabcd})}{\ge}\Vert \mathcal{P}_{T^{\perp}} W\Vert_{\ast}\\
	&=\Vert \mathcal{P}_{T^{\perp}} Z \Vert_{\ast}
	\end{align*}
	and hence Lemma \ref{lemma:descentconechar} entails that $ Z \in \overline{  \dcone{h_0 m^*_0}  } $.  Moreover, note that by the triangle inequality and by the assumption $ L \le \frac{1}{36} KN $ it holds that
	\begin{equation}\label{ineq:Zupperbound}
	\begin{split}
	\Vert Z \Vert_F  \ge& \Vert W \Vert_F - \beta\\
	=& 1 -3\sqrt{\frac{L}{KN}}\\
	\ge&\frac{1}{2}.
	\end{split}
	\end{equation}
	These observations together with $ \mathcal{A} \left( W \right)=0  $ yield that
	\begin{equation*}
	\Vert \mathcal{A} \left(Z\right) \Vert = \Vert \mathcal{A} \left( \beta h_0 m^*_0  \right) \Vert \overset{(\ref{ineq:event1})}{<} 2\beta = 6 \sqrt{ \frac{L}{KN} } \overset{(\ref{ineq:Zupperbound})}{\le} 12 \sqrt{ \frac{L}{KN} } \Vert Z \Vert_F.
	\end{equation*}
	This shows (\ref{ineq:intern12345}), as desired.
	
\end{proof}

\subsection{Upper bound for matrix completion}\label{section:proof1MC}
In this section, we prove Proposition \ref{theorem:mainprop}. For that we introduce sets $ \mathcal{N}_a $, $ a \in \left[n_1\right]$,  via
\begin{equation*}
\mathcal{N}_{a}:= \left\{ b \in \left[ n_2 \right]: a=a_i \text{ and }  b=b_i \text{ for some } i \in \left[m\right]   \right\}.
\end{equation*}
That is, $ \mathcal{N}_a$ contains all the indices of the $a$th row of the matrix $X_0$, which are observed by the measurements. Furthermore, define by $P_{\mathcal{N}_a} \in \mathbb{R}^{n_2 \times n_2} $ the projection onto the coordinates, which are contained in $\mathcal{N}_a$, i.e. $ P_{\mathcal{N}_a} = \sum_{b \in \mathcal{N}_a} e_b e^*_b $. By $ P_{\mathcal{N}^{\perp}_a} = \sum_{b \in \left[n_2\right] \setminus \mathcal{N}_a} e_b e^*_b $ we denote the coordinate projection onto $ \left[n_2\right] \setminus \mathcal{N}_a $. \\

We need the following technical lemma.
\begin{lemma}\label{lemma:auxiliary1}
	Let $ V\in \mathbb{R}^{n_2 \times r} $ be an isometry, i.e. $V^*V= \Id $. Assume that 
	\begin{equation}\label{ineq:assumption23}
	m \ge C_1 rn_1 \mu^2 \left(V\right) \log (2r).
	\end{equation}
	Then with probability at least $1 - \mathcal{O} \left(\exp \left( -\frac{m}{C_2 r \mu^2 \left(V\right) } \right) \right) $ there exists $a\in \left[n_1\right] $ such that
	\begin{equation*}
	\Vert  P_{\mathcal{N}_a} V  \Vert \le \sqrt{ \frac{2m}{n_1 n_2}}.
	\end{equation*}
	$C_1$ and $C_2$ are universal constants.
\end{lemma}

\begin{proof}
	For each $ a\in \left[n_1\right] $ we set $\mathcal{I}_a := \left\{  i\in \left[ m \right] : a_i =a   \right\}$ and define the event
	\begin{equation*}
	\mathcal{E}_a := \left\{ \Vert  P_{\mathcal{N}_a} V  \Vert^2  \le \frac{2m}{n_1 n_2} \right\}.
	\end{equation*}
	We will first derive a lower bound for $\mathbb{P} \left( \mathcal{E}_a  \big\vert \mathcal{I}_a \right)  $. For that we note that $\Vert  P_{\mathcal{N}_a} V  \Vert^2 = \Vert V^*  P_{\mathcal{N}_a} V  \Vert $. Let $v_1, v_2, \ldots, v_{n_2} $ denote the rows of the matrix $V$. By definition of $\mathcal{N}_a$ and $ \mathcal{I}_a $ it follows that for $X_i := v_{b_i} v^*_{b_i}$
	\begin{equation}\label{ineq:intern99}
	V^*  P_{\mathcal{N}_a} V = \sum_{b \in \mathcal{N}_a} v_b v^*_b \preceq \sum_{i \in \mathcal{I}_a} v_{b_i} v^*_{b_i} = \sum_{i \in \mathcal{I}_a} X_i .
	\end{equation}
	Here we write $A \preceq B $ for two symmetric matrices $A$ and $B$, if and only if $B-A$ is positive semidefinite. 
	By (\ref{ineq:intern99}) it is sufficient to bound the probability of the event
	\begin{equation}\label{ineq:internxx23}
	\left\{  \Big\Vert \sum_{i\in \mathcal{I}_a} X_{i} \Big\Vert \le \frac{2m}{n_1 n_2} \right\}  \subset \mathcal{E}_a
	\end{equation}
	conditionally on $\mathcal{I}_a $. To bound $\Vert \sum_{i \in \mathcal{I}_a} X_{i} \Vert $ we will use the matrix Bernstein inequality (see, e.g., \cite[Theorem 6.1.1]{tropp2015introduction}) conditionally on $ \mathcal{I}_a $, which requires as ingredients
	\begin{equation*}
	\mathbb{E} \left[ \sum_{i \in \mathcal{I}_a} X_{i} \ \big\vert \mathcal{I}_a \right] = \frac{\vert \mathcal{I}_a \vert}{n_2} \sum_{b=1}^{n_2} v_b v^*_b=  \frac{\vert \mathcal{I}_a \vert}{n_2} V^*V =  \frac{\vert \mathcal{I}_a \vert}{n_2} \Id,
	\end{equation*}
	an upper bound for $\sigma^2 \left(\mathcal{I}_a\right) := \Big\Vert \mathbb{E} \left[ \sum_{i\in \mathcal{I}_a}  \left(   X_{i} - \mathbb{E} \left[ X_{i} \right]   \right)^2   \big\vert \mathcal{I}_a \right] \Big\Vert$ and a constant $K>0$ such that $\Vert X_i - \mathbb{E} \left[X_i\right] \Vert \le K $ almost surely. To bound $\sigma^2 \left(\mathcal{I}_a\right)$ we note that
	\begin{align*}
	\mathbb{E} \left[ \sum_{i \in \mathcal{I}_a}  \left(   X_{i} - \mathbb{E} \left[ X_{i} \right]   \right)^2 \big\vert \mathcal{I}_a  \right] &\preceq   \mathbb{E} \left[ \sum_{i \in \mathcal{I}_a} X^2_{i}  \big\vert \mathcal{I}_a  \right]\\
	&\preceq \left(  \underset{b \in \left[n_2\right]}{\max} \Vert v_b \Vert^2 \right)  \mathbb{E}  \left[ \sum_{i \in \mathcal{I}_a}  X_{i}  \big\vert \mathcal{I}_a  \right]\\
	&=  \frac{ \vert \mathcal{I}_a \vert  \left(  \underset{b \in \left[n_2\right]}{\max} \Vert v_b \Vert^2 \right)}{ n_2}  \Id \\
	&= \frac{ \vert \mathcal{I}_a \vert \mu^2 \left(V\right) r}{n^2_2} \Id,
	\end{align*}
	where the fourth line is due to the definition of $\mu^2 \left(V\right) $. This implies that 
	\begin{equation*}
	\sigma^2 \left( \mathcal{I}_a \right) \le  \frac{ \vert \mathcal{I}_a \vert \mu^2 \left(V\right) r}{n^2_2}.
	\end{equation*}
	To find an appropriate $K>0$ note that almost surely
	\begin{align*}
	\big\Vert X_{i} - \mathbb{E} X_{i} \big\Vert &= \big\Vert X_{i} - \frac{1}{n_1n_2} \Id \big\Vert\\ 
	&\le \underset{i \in \left[n_2\right]}{\max} \big\Vert v_i v^*_i - \frac{1}{n_1n_2} \Id \big\Vert\\
	&\le  \frac{1}{n_1 n_2} + \underset{i \in \left[n_2\right]}{\max} \big\Vert v_i \big\Vert^2\\
	&\le \left( \frac{1}{n_1r}    +1 \right) \underset{i \in \left[n_2\right]}{\max}\big\Vert v_i \big\Vert^2\\
	&\le \frac{2r}{n_2} \mu^2 \left(V\right)=:K,
	\end{align*}
	where in the fourth line we used that $  \underset{i \in \left[n_2\right]}{\max} \Vert v_i \Vert^2 \ge \frac{r}{n_2}  =  \frac{1}{n_2} \sum_{i \in \left[n_2\right]} \Vert v_i \Vert^2 $. 
	Finally to apply Bernstein inequality we need that the $X_i$'s are independent conditionally on $ \mathcal{I}_a$, which follows from the fact that the $a_i$'s and $b_i$'s are drawn independently. With these ingredients the matrix Bernstein inequality yields that
	\begin{align*}
	\mathbb{P} \left(  \Big\Vert  \sum_{i\in \mathcal{I}_a}  X_{i}  - \frac{\vert \mathcal{I}_a \vert}{ n_2} \Id \Big\Vert \le t  \ \Big\vert  \mathcal{I}_a  \right) &\ge 	1- 2r \exp \left( -c \min \left\{\frac{t^2}{\sigma^2 \left( \mathcal{I}_a \right) }; \frac{t}{K} \right\}   \right)\\
	&\ge 1-  2r\exp \left( -c \min  \left\{  \frac{ n^2_2 t^2}{ \vert \mathcal{I}_a \vert r \mu^2 \left( V \right) } ; \frac{n_2 t}{2r \mu^2 \left(V\right)}    \right\}   \right).
	\end{align*}
	Setting $t= \frac{m}{2n_1 n_2} $ this implies that for fixed $ a\in \left[n_1\right] $ it holds that
	\begin{equation}\label{ineq:intern47}
	\begin{split}
	\mathbb{P} \left(  \Big\Vert  \sum_{i\in \mathcal{I}_a}  X_{i}  \Big\Vert \le \frac{m}{2n_1 n_2}  + \frac{\vert \mathcal{I}_a \vert}{n_2}   \Big\vert \ \mathcal{I}_a \right) &\ge \mathbb{P} \left( \Big\Vert  \sum_{i\in \mathcal{I}_a}  X_{i}  - \frac{\vert \mathcal{I}_a \vert}{ n_2} \Id \Big\Vert \le \frac{m}{2n_1 n_2}  \ \Big\vert \ \mathcal{I}_a \right)   \\
	&\ge 1 - 2r\exp \left( - \frac{cm}{4 r n_1 \mu^2 \left( V \right) }  \min  \left\{  \frac{ m}{  \vert \mathcal{I}_a \vert n_1  } ; 1   \right\}  \right).
	\end{split}
	\end{equation}
	To complete the proof we restrict our attention to $ A:= \left\{  a\in \left[n_1\right]: \ \vert \mathcal{I}_a \vert  \le \frac{4m}{3n_1}    \right\} $ as it follows from (\ref{ineq:internxx23}) that
	\begin{equation}\label{ineq:intern48}
	\left\{  \vert \mathcal{I}_a \vert \le \frac{4m}{3n_1}   \right\} \cap \left\{ \Big\Vert  \sum_{i\in \mathcal{I}_a}  X_{i}  \Big\Vert \le \frac{m}{2n_1 n_2}  + \frac{\vert \mathcal{I}_a \vert}{n_2}    \right\}    \subset \mathcal{E}_a,
	\end{equation}
	and, consequently, for $a \in A $ we obtain that
	\begin{equation}\label{ineq:interns8}
	\begin{split}
	\mathbb{P} \left( \mathcal{E}_a   \big\vert \mathcal{I}_a \right) &\overset{(\ref{ineq:intern48}), (\ref{ineq:intern47})}{\ge} 1 - 2r\exp \left( - \frac{cm}{4 r n_1 \mu^2 \left( V \right) }  \min  \left\{  \frac{ m}{  \vert \mathcal{I}_a \vert n_1  } ; 1   \right\}  \right)\\
	& \overset{a \in A}{\ge}  1 - 2r\exp \left( - \frac{3cm}{16 r n_1 \mu^2 \left( V \right) }   \right).
	\end{split}
	\end{equation}
	As the $ \mathcal{E}_a $'s only depend on $ \left\{ b_i \right\}_{i \in \mathcal{I}_a } $ and are hence independent conditionally on $ \mathcal{I}_{\alpha} $, this implies that
	\begin{align*}
	\mathbb{P} \left( \underset{a\in A  }{\bigcap} \mathcal{E}^c_a   \text{ } \Big\vert \left\{  \mathcal{I}_a  \right\}^{n_1}_{a=1}    \right) &=  \prod_{a \in A}   	\mathbb{P} \left(  \mathcal{E}^c_a   \ \Big\vert \left\{  \mathcal{I}_a  \right\}^{n_1}_{a=1}   \right)   \\
	&= \prod_{a \in A}   	\mathbb{P} \left(  \mathcal{E}^c_a   \ \Big\vert   \mathcal{I}_a    \right)\\
	& \overset{(\ref{ineq:interns8})}{\le}  \prod_{a \in A}  \left(  2r\exp \left( - \frac{3cm}{16 r n_1 \mu^2 \left( V \right) } \right) \right)\\
	& =  \left(  2r\exp \left( - \frac{3cm}{16 r n_1 \mu^2 \left( V \right) } \right) \right)^{\vert A \vert}\\
	&= \left( \exp \left( \log \left(2r\right) - \frac{3cm}{16rn_1 \mu^2 \left(V\right)}   \right)    \right)^{\vert A \vert}\\
	&\le \left(  \exp \left(   - \frac{m}{C_1' r n_1 \mu^2 \left(V\right)}   \right)    \right)^{\vert A \vert},
	\end{align*}
	where in the last line we have used assumption (\ref{ineq:assumption23}) with $C_1$ large enough. Furthermore, note that
	\begin{equation*}
	m = \sum_{a=1}^{n_1} \vert \mathcal{I}_a \vert \ge \sum_{a \in \left[n_1\right] \setminus A} \vert \mathcal{I}_a \vert \ge   \frac{4m}{3n_1} \big\vert \left[n_1\right] \setminus A \big\vert = \frac{4m}{3n_1} \left( n_1 - \vert A \vert  \right)
	\end{equation*}
	implies that $ \vert A \vert \ge \frac{n_1}{4} $ almost surely. Hence, it follows that
	\begin{align*}
	\mathbb{P} \left( \underset{a\in A  }{\bigcap} \mathcal{E}^c_a   \text{ } \Big\vert \left\{  \mathcal{I}_a  \right\}^{n_1}_{a=1}    \right)&\le  \exp\left(  -\frac{m}{C_2 r \mu^2 \left(V\right)}  \right).
	\end{align*}
	This shows that conditional on $\left\{  \mathcal{I}_a \right\}^{n_1}_{a=1} $ we have that almost surely
	\begin{equation*}
	\mathbb{P} \left( \underset{a\in \left[n_1\right]  }{\bigcup} \mathcal{E}_a \\  \\ \\ \Big\vert \left\{ \mathcal{I}_a  \right\}^{n_1}_{a=1}   \right) \ge 	\mathbb{P} \left( \underset{a\in A  }{\bigcup} \mathcal{E}_a    \\  \\ \\ \Big\vert \left\{  \mathcal{I}_a  \right\}^{n_1}_{a=1}     \right)     \ge 1-  \exp\left(  -\frac{m}{C_2 r \mu^2 \left(V\right)}  \right).
	\end{equation*}
	Taking expectations yields the claim.

	
\end{proof}

Now we are prepared to give a proof of Proposition \ref{theorem:mainprop}.
\begin{proof}[Proof of Proposition \ref{theorem:mainprop}]
	For the proof we will condition on two events $\mathcal{E}_1$ and $\mathcal{E}_2 $, which we will define in the following. The event $\mathcal{E}_1 $ is defined by
	\begin{equation}\label{ineq:eventxxx1}
	\mathcal{E}_1 := \left\{ \Vert \mathcal{A} \left(UV^*\right) \Vert^2 \le 2\Vert UV^* \Vert^2_F = 2r \right\}.
	\end{equation}
	Observe that
	\begin{align*}
	\Vert \mathcal{A} \left(UV^*\right) \Vert^2 = \frac{n_1 n_2}{m} \sum_{i=1}^{m} \vert \left(UV^*\right)_{a_i,b_i} \vert^2 = \sum_{i=1}^{m} X_i,
	\end{align*}
	where we have set $ X_i := \frac{n_1 n_2}{m}  \vert \left(UV^*\right)_{a_i,b_i} \vert^2  $. Note that one has almost surely that
	\begin{align*}
	 X_i &= \frac{n_1 n_2}{m} \left(\sum_{k=1}^{r}  U_{a_i,k} V_{b_i,k} \right)^2\\
	 &\le \frac{n_1 n_2}{m} \left(  \sum_{k=1}^{r} \vert U_{a_i,k}  \vert^2  \right) \left( \sum_{k=1}^{r} \vert V_{b_i,k} \vert^2  \right)\\
	 &\le  \frac{\mu^2 \left( U \right) \mu^2 \left( V \right) r^2 }{m}
	\end{align*}
	where we have applied Cauchy-Schwarz and the definition of $\mu \left(U\right) $ and $ \mu \left(V\right) $. Hence, one has 
	\begin{align*}
	\mathbb{E} X_i &= \frac{\Vert UV^* \Vert_F^2}{m}= \frac{r}{m},\\
	\mathbb{E} X^2_i  &\le  \frac{\mu^2 \left( U \right) \mu^2 \left( V \right) r^2 }{m} \mathbb{E}X_i     = \frac{r^3 \mu^2 \left( U \right) \mu^2 \left( V \right) }{m^2}
	\end{align*}
	where we used $ \Vert UV^* \Vert^2_F  =r$ and the previous estimate. Hence, by the Bernstein inequality  (see, e.g., \cite[Theorem 2.8.4]{vershynin2016high}) we obtain that
	\begin{equation*}
	\mathbb{P} \left( \Big\vert \sum_{i=1}^{m} X_i - \Vert UV^* \Vert^2_F  \Big\vert \ge t \right) \le 2   \exp \left( -c \min \left\{  \frac{t^2 m^2}{r^3 \mu^2 \left( U \right) \mu^2 \left( V \right) }   ; \frac{tm}{r^2  \mu^2 \left( U \right) \mu^2 \left( V \right) }   \right\}  \right).
	\end{equation*}
	By setting $ t= \Vert UV^* \Vert^2_F =r $ we observe that $ \mathcal{E}_1 $ holds with probability at least $1- 2\exp \left( -\frac{cm}{r \mu^2 \left( U \right) \mu^2 \left( V \right)  }  \right) $.
	
	The second event $\mathcal{E}_2 $ is defined by
		\begin{equation*}
		\mathcal{E}_2 := \left\{ \exists a \in \left[n_1\right] \text{ such that }   \Vert P_{\mathcal{N}_a}  V \Vert \le \sqrt{ \frac{2m}{n_1 n_2}} \right\}.
		\end{equation*}
For $C_1$ in assumption (\ref{ineq:assumptionbla2})  chosen large enough Lemma \ref{lemma:auxiliary1} then entails that $\mathbb{P} \left( \mathcal{E}_2 \right) \ge 1 - \mathcal{O} \left(\exp \left( -\frac{cm}{  r \mu^2 \left(V\right) } \right) \right) $. 
 Consequently, we can find $ a \in \left[n_1\right]$ (depending on the random sampling pattern) such that the condition defining $ \mathcal{E}_2 $ is satisfied.\\
	
	\noindent Note that in order to prove Proposition \ref{theorem:mainprop} it is enough to find $ Z \in \overline{  \dcone{X_0}  } \setminus \left\{ 0 \right\} $ such that
	\begin{equation}\label{ineq:internxxxx}
	\Vert \mathcal{A} \left(Z\right) \Vert < 8  \sqrt{\frac{m}{r n_1 n_2}} \Vert Z \Vert_F,
	\end{equation}
	because by definition of the closure and the continuity of $ \mathcal{A} $ this implies that there is a matrix $\tilde{Z} \in \dcone{X_0} \setminus \left\{ 0 \right\} $ such that
	\begin{equation*}
	\frac{\Big\Vert \mathcal{A} \left( \tilde{Z}  \right) \Big\Vert}{\Vert \tilde{Z}  \Vert_F} \le 8 \sqrt{\frac{m}{rn_1 n_2}},
	\end{equation*}
	which implies (\ref{ineq:minconvalueMC}) with constant $ C_3 =8 $. In the following we will construct such a matrix $Z$. Let $x \in \mathbb{R}^r $ be a vector such that $ \Vert x \Vert=1 $. Then for $  a\in \left[n_1\right]$ as above we define the vector $ w_a \in \mathbb{R}^{n_2} $ by 
	\begin{equation*}
	w_a:=  P_{\mathcal{N}^{\perp}_a} Vx
	\end{equation*}
	and set
	\begin{equation*}
	W:= - \frac{\innerproduct{ e_a w_a^* , UV^*}_F}{\vert \innerproduct{ e_a w_a^* , UV^*}_F \vert}   e_a w_a^*= - \frac{\innerproduct{ e_a w_a^* , UV^*}_F}{\vert \innerproduct{ e_a w_a^* , UV^*}_F \vert}  e_a x^* V^* P_{\mathcal{N}^{\perp}_a}.
	\end{equation*}
	It follows directly from the definition of $\mathcal{N}_a$ that $ \mathcal{A} \left(W\right) =0  $. In the following let $T$ be the tangent space of the manifold of rank-$r$ matrices at $X_0$ as defined in (\ref{equ:deftangentspace}). Furthermore, denote by $P_{U}=UU^*$ the orthogonal projection onto the column space of $U$ and, analogously, by $P_{V}=VV^*$ the orthogonal projection onto the column space of $V$. Then we obtain that  
	\begin{align*}
	\Vert \mathcal{P}_{T^{\perp}} W \Vert_{\ast} &= 	\Vert \mathcal{P}_{T^{\perp}} \left( e_a w_a^* \right) \Vert_{\ast}\\
	&= 	\Vert P_{U^{\perp}} e_a w_a^* P_{V^{\perp}} \Vert_{\ast}\\
	&=\Vert P_{U^{\perp}} e_a w_a^* P_{V^{\perp}} \Vert_{F}\\
	&= \Vert P_{U^{\perp}} e_a  \Vert \Vert P_{V^{\perp}}  w_a \Vert\\
	&\le \Vert  P_{V^{\perp}}  w_a \Vert,
	\end{align*}
	where in the second equality we have used that $ \mathcal{P}_{T^{\perp}} M= P_{U^{\perp}} M P_{V^{\perp}} $ for all $M\in \mathbb{R}^{n_1 \times n_2} $ and in the last line we used that $ \Vert P_{U^{\perp}} e_i \Vert \le 1 $. Plugging in $w_a =  P_{\mathcal{N}^{\perp}_a} Vx$ it follows that
	\begin{align*}
	\Vert \mathcal{P}_{T^{\perp}} W \Vert_{\ast}&\le  \Vert  P_{V^{\perp}}  P_{\mathcal{N}^{\perp}_a} Vx \Vert\\
	&=  \Vert  P_{V^{\perp}}  P_{\mathcal{N}_a} Vx \Vert,
	\end{align*}
	where the last line is due to $ P_{V^{\perp}}  Vx  =0 $. The fact that $\Vert P_{V^{\perp}}  \Vert \le 1$ then yields that
	\begin{equation}\label{ineq:internsad3}
	\begin{split}
	\Vert \mathcal{P}_{T^{\perp}} W \Vert_{\ast} &\le \Vert P_{\mathcal{N}_a} Vx \Vert\\
	&\le \Vert P_{\mathcal{N}_a} V \Vert \Vert x \Vert\\
	&\le \sqrt{\frac{2m}{n_1 n_2}}
	\end{split}
	\end{equation}
	where for the last line we used that $ a \in \left[n_1\right] $ was chosen such that the condition in $\mathcal{E}_2 $ holds. This shows that $W$ is relatively close to $T$. Based on $W$ we now aim to find $Z\in \overline{   \dcone{X}  } $ of the form
	\begin{equation*}
	Z:= W-\beta UV^*,
	\end{equation*}
	where $ \beta >0 $ will be chosen in the following such that $Z \in \overline{   \dcone{X_0}   } $, which by Lemma \ref{lemma:descentconechar} is equivalent to $ -  \innerproduct{ UV^*, Z}_F \ge  \Vert \mathcal{P}_{T^{\perp}} Z \Vert_{\ast} $ . First, we note that
	\begin{equation}\label{ineq:internsad1}
	\Vert \mathcal{P}_{T^{\perp}} Z \Vert_{\ast} = \Vert  \mathcal{P}_{T^{\perp}} W  \Vert_{\ast} \le \sqrt{\frac{2m}{n_1 n_2}}
	\end{equation}
	due to $ \mathcal{P}_{T^{\perp}}  \left(UV^*\right) =0 $ and the inequality chain (\ref{ineq:internsad3}). Furthermore, we have that
	\begin{equation}\label{ineq:internsad2}
	\begin{split}
	-\innerproduct{UV^*,Z}_F &= r \beta - \innerproduct{UV^*,W}_F \\
	&= r \beta  + \innerproduct{UV^*, \frac{\innerproduct{ e_a w_a^* , UV^*}_F}{\vert \innerproduct{ e_a w_a^* , UV^*}_F \vert} e_a w_a^* }_F \\
	&= r \beta + \vert \innerproduct{ e_a w_a^* , UV^*}_F \vert   \\
	&\ge r \beta.
	\end{split}
	\end{equation}
	Hence, setting $\beta = 2 \sqrt{\frac{m}{r^2 n_1 n_2}}  $ and combining (\ref{ineq:internsad1}) and (\ref{ineq:internsad2}) it follows that $Z \in  \overline{   \dcone{X_0}   }$. This $Z$ also satisfies (\ref{ineq:internxxxx}). To see that we observe that
	\begin{align*}
	\Vert Z \Vert_F &\ge \Vert e_a w_a^* \Vert_F - \beta \Vert UV^* \Vert_F\\
	& = \Vert w_a \Vert - \beta \sqrt{r} \\
	& =  \Vert P_{\mathcal{N}^{\perp}_a} Vx \Vert -\beta \sqrt{r} \\
	&= \sqrt{ \Vert Vx \Vert^2 -  \Vert  P_{\mathcal{N}_a} Vx \Vert^2 } -\beta \sqrt{r} \\
	&\ge \sqrt{1 - \frac{2m}{n_1 n_2}} - 2 \sqrt{ \frac{m}{r n_1  n_2}} \\
	&>  \frac{1}{2},
	\end{align*}
	where in the last line we used the assumption that $ m \le \frac{n_1 n_2}{32} $. Furthermore, from $ \mathcal{A} \left(W\right)=0 $ it follows that
	\begin{equation*}
	\Vert \mathcal{A} \left(Z\right) \Vert = \beta \Vert \mathcal{A} \left( UV^*  \right) \Vert \overset{(\ref{ineq:eventxxx1})}{\le}  \beta \sqrt{2r} = 2 \sqrt{\frac{2m}{r n_1 n_2}} < 8 \sqrt{\frac{m}{r n_1 n_2}} \Vert Z \Vert_F .
	\end{equation*}
	Combining the last two inequality chains implies (\ref{ineq:internxxxx}), which completes the proof.
\end{proof}

\subsection{Proof of Theorem \ref{thm:mainresult} and Theorem \ref{thm:mainresultMC}}
As already mentioned Theorem \ref{thm:mainresult} can be deduced from Proposition \ref{theorem:mainpropBD} and Theorem \ref{thm:mainresultMC} can be deduced from Proposition \ref{theorem:mainprop}. We only show how to prove Theorem \ref{thm:mainresult}, as the proof of Theorem \ref{thm:mainresultMC} is analogous.
\begin{proof}[Proof of Theorem \ref{thm:mainresult}]
	By Proposition \ref{theorem:mainpropBD} and the definition of the minimum conic singular value $ \lambda_{\min} \left(\mathcal{A},  \dcone{h_0 m^*_0} \right) $ with probability at least $1 - \mathcal{O} \left(\exp \left( -\frac{K}{C_2 \mu^2 } \right) \right)$  there is a matrix $ Z \in  \dcone{h_0 m^*_0} \setminus \left\{ 0 \right\} $ such that 
	\begin{equation}\label{ineq:intern12}
	\Vert \mathcal{A} \left(Z\right) \Vert \le C_3 \sqrt{\frac{L}{KN}} \Vert Z \Vert_F .
	\end{equation}
	and such that $\tilde{X}_t := h_0 m^*_0 + tZ $ obeys $ \Vert \tilde{X}_t \Vert_{\ast} \le \Vert h_0 m^*_0   \Vert_{\ast} $ for all $ 0 < t \le 1 $. Next, set $e_t= \frac{t}{2} \mathcal{A} \left(Z\right) $. Then for $y_t= \mathcal{A} \left( h_0 m^*_0 \right) +e_t $ we have that
	\begin{equation*}
	\Vert \mathcal{A} \left( \tilde{X}_t \right) -\tilde{y}_t \Vert = \Vert \mathcal{A} \left(tZ\right) -e \Vert = \frac{t}{2} \Vert \mathcal{A} \left(Z\right) \Vert.
	\end{equation*}
	Hence, by setting $\tau_0 :=  \frac{ \Vert \mathcal{A} \left(Z\right) \Vert}{2}$ we observe that $\Vert \mathcal{A} \left(\hat{X}_t\right) -y_t \Vert = t\tau_0 $. Furthermore, note that
	\begin{align*}
	\Vert \tilde{X}_t - X_0  \Vert_F = \Vert tZ \Vert_F =\overset{(\ref{ineq:intern12})}{\ge} \frac{t}{C_3} \sqrt{\frac{KN}{L}}  \Vert \mathcal{A} \left(Z\right) \Vert = \frac{ 2 t \tau_0}{C_3} \sqrt{\frac{KN}{L}}
	\end{align*}
	Now let $0 < \tau \le \tau_0 $. Then by setting $t=\frac{\tau}{\tau_0}$, $\tilde{X}:= \tilde{X}_t $, $y:=y_t $, and $e:= e_t$ the desired claim follows.
\end{proof}

\section{Stability of blind deconvolution}\label{section:stability}

\subsection{Outline of the proof and main ideas}
The goal of this section is to prove Theorem \ref{thm:stabilityBD}. We first give a proof sketch and present the key ideas.
We have seen in Proposition \ref{theorem:mainpropBD} that for certain isometries $ B \in \mathbb{C}^{L \times K} $ with high probability one has that $ \lambda_{\min} \left( \mathcal{A}, \dcone{h_0 m^*_0} \right) \lesssim \sqrt{\frac{L}{KN}}  $. Hence, applying Theorem \ref{thm:chandrasekaran} cannot lead to very strong error estimates. However, if we closely inspect the proof of Proposition \ref{theorem:mainpropBD} we observe the following. Again, denote by $T$ the tangent space of the manifold of rank-$1$ matrices at point $ h_0 m^*_0 $ as defined in (\ref{equ:deftangentspace}) and assume that $\Vert h_0 \Vert = \Vert m_0 \Vert = 1 $. By construction we have that $Z= W - \beta h_0 m^*_0 $, where $W= \frac{\innerproduct{ h_0,e_i }}{\Vert m_i^{\perp} \Vert \vert \innerproduct{h_0,e_i} \vert} e_i \left(m_i^{\perp}\right)^* $.  This implies that
\begin{align*}
\vert \innerproduct{Z,h_0 m^*_0}_F \vert &\le \Big\vert  \innerproduct{\frac{\innerproduct{ h_0,e_i }}{\Vert m_i^{\perp} \Vert \vert \innerproduct{h_0, e_i} \vert} e_i \left(m_i^{\perp}\right)^* , h_0 m_0^*}_F \Big\vert    + \beta \Vert h_0 m^*_0 \Vert^2_F\\
&= \vert \innerproduct{e_i, h_0} \vert \Vert m^{\perp}_i \Vert   + \beta \Vert h_0 m^*_0 \Vert^2_F\\
&\le \frac{\mu}{\sqrt{L}} + \beta \lesssim \frac{\mu}{\sqrt{L}} + \sqrt{\frac{L}{KN}},
\end{align*}
where we have used the triangle inequality in the first line and the definition of $m^{\perp}_i$ in the second line. In the third line we used that $ \Vert m^{\perp}_i \Vert\le 1 $, $ \vert \innerproduct{e_i, h_0} \vert \le \frac{\mu}{\sqrt{L}}$, and $ \beta = \sqrt{\frac{L}{KN}} $. As $\Vert Z \Vert_F \gtrsim 1 $ this implies that $ \frac{\innerproduct{ Z, h_0 m^*_0}_F}{\Vert Z \Vert_F \Vert h_0 m^*_0 \Vert_F }$ is quite small, meaning that $Z$ and $h_0 m^*_0$ are almost orthogonal to each other. All the descent directions $Z$ with this property, however, have in common that the admissible descent step size 
\begin{equation*}
t_0 :=  \max \left\{ t>0 : \Vert h_0 m^*_0 + tZ \Vert_{\ast} \le \Vert h_0 m^*_0 \Vert_{\ast} \right\}  
\end{equation*}
is necessarily very small. Geometrically this corresponds to the fact that the nuclear norm ball is curved near $X_0$, which is why its near-tangential behaviour only holds locally. This will be made precise in Lemma \ref{lemma:maximaldescent} below. 
For this reason the idea of the proof of Theorem \ref{thm:stabilityBD} is to split the descent cone into two parts. One part will consist of all the matrices aligned with $h_0 m^*_0$. The second part will consist of all remaining matrices, which are almost orthogonal to $h_0 m^*_0$. As mentioned above, these matrices must necessarily be close to $T$. The first part is captured by the set $ \mathcal{E}_{\mu,\delta} $ with
\begin{equation*}
\mathcal{E}_{\mu,\delta}:= \underset{h_0 \in \mathcal{H_{\mu}},m_0 \in \mathbb{C}^N}{\bigcup} \left\{  Z\in  \dcone{h_0 m^*_0} :  \delta \le  \frac{- \text{Re} \left( \innerproduct{Z, h_0 m^*_0}_F \right)}{\Vert h_0 m^*_0 \Vert_F} \text{ and } \Vert Z \Vert_F =1  \right\},
\end{equation*}
where $\delta >0$. In Section \ref{section:conicsingular} we will show that with high probability it holds that
\begin{equation}\label{ineq:xwu}
\underset{Z \in \mathcal{E}_{\mu,\delta}}{\inf} \Vert \mathcal{A} \left(Z\right) \Vert \gtrsim \frac{\delta^2}{\log^{2} (L)  \mu^2}.
\end{equation}
Hence, if we have for the minimizer $\hat{X}$ of (\ref{opt:SDP}) that $ \hat{X}-h_0m^*_0 $ is an element of the conic hull of $ \mathcal{E}_{\mu,\delta} $ we can proceed similarly as in \cite{chandrasekaran2012convex} to obtain near-optimal error bounds. Let us briefly explain which property of $ \mathcal{E}_{\mu,\delta} $ allows us to show (\ref{ineq:xwu}). We define for any matrix $W \in \mathbb{C}^{K\times N} $ its $\Bnorm{\cdot} $-norm by
\begin{align*}
\Bnorm{W} &:=  \sum_{\ell =1}^{L} \Vert W^* b_{\ell}   \Vert.
\end{align*}
In other words $ \Bnorm{W} $ is the $ \ell_1$-norm of the vector $ \left(  \Vert W^* b_{\ell} \Vert  \right)^L_{\ell=1} $. We show in Lemma \ref{lemma:onenormdescentcone} below that all $Z \in \mathcal{E}_{\mu, \delta}$ have rather large $\Vert \cdot \Vert_{B_1}$-norm, which entails that the mass of the vector $\left(  \Vert W^* b_{\ell} \Vert  \right)^L_{\ell=1}$ cannot be concentrated on only very few entries. This in turn will allow us to employ a non-i.i.d. version of Mendelson's small ball method \cite{koltchinskii2015bounding}, allowing us to show a lower bound for (\ref{ineq:xwu}), see Lemma \ref{lemma:deconvconicbound} below.
To understand the behaviour on the second part recall from Proposition \ref{theorem:mainpropBD} that for matrices $Z/\Vert Z \Vert_F \in  \dcone{h_0 m^*_0}  \setminus \mathcal{E}_{\mu,\delta}$ the quantity $\Vert \mathcal{A} \left(Z\right) \Vert $ may be quite small, so a uniform bound is not feasible. However, as $Z$ is almost orthogonal to $h_0 m^*_0 $,  also $\Vert \mathcal{P}_{T^{\perp} }Z \Vert_{\ast} $ must be rather small because of the characterization of the descent cone, Lemma \ref{lemma:descentconechar}. Hence, $Z$ is close to the tangent space and is almost orthogonal to $h_0 m^*_0 $. For that reason, whenever $\Vert h_0m^*_0 + t Z \Vert_{\ast} \le \Vert h_0 m^*_0  \Vert_{\ast} $ holds, the cylindrical shape of the nuclear norm ball implies that $t>0$ is small. This fact is captured by Lemma \ref{lemma:maximaldescent} below. Theorem \ref{theorem:mainpropBD} can then be proven by combining inequality (\ref{ineq:xwu}) and  Lemma \ref{lemma:maximaldescent}, see Section \ref{section:proofmain}.

\subsection{A lower bound for the minimum conic singular value}\label{section:conicsingular}
First we recall the notion of Gaussian width (see, e.g., \cite{vershynin2016high}).
\begin{definition}\label{def:Gaussianwidth}
	For a set $ \mathcal{E} \subset \mathbb{C}^{K \times N} $ its Gaussian width is defined by
	\begin{equation*}
	\omega \left( \mathcal{E}  \right) := \mathbb{E} \left[ \underset{X  \in \mathcal{E}}{\sup} \ \text{Re} \left(  \innerproduct{X,G}_F  \right) \right],
	\end{equation*}
	where $ G \in \mathbb{C}^{K \times N} $ is a matrix, whose entries are independent and identically distributed random variables with distribution $ \mathcal{CN} \left(0,1\right) $.
\end{definition}
This definition allows us to state the following lemma, which is important for our analysis of the conic singular value. It relies on a uniform lower bound on the number of measurements whose magnitude is larger than a certain constant and is a variant of Theorem 2.1 in \cite{koltchinskii2015bounding}.
\begin{lemma}\label{lemma:smallballmethod}
	Let $\mathcal{E} \subset \mathbb{C}^{K \times N}$ be a symmetric set, i.e., $ \mathcal{E} = - \mathcal{E} $. For all $ \xi >0 $ and $t>0$ it holds with probability at least $ 1- \exp \left(-2t^2\right) $ that
	\begin{align}\label{ineq:intern333333}
	&\underset{X \in \mathcal{E}}{\inf}  \Big\vert \left\{ \ell \in    \left[L\right]: \  \vert \langle b_{\ell} c^*_{\ell}, X \rangle_F \vert \ge \xi  \right\}   \Big\vert
	\ge  \underset{X \in \mathcal{E}}{\inf} \left( \sum_{\ell=1}^{L} \mathbb{P} \left( \vert \langle b_{\ell} c^*_{\ell}, X \rangle_F \vert \ge  2 \xi  \right)   \right)    - \frac{4 \omega \left( \mathcal{E} \right)}{\xi}  -  t  \sqrt{L}.
	\end{align}
	Here $ \varepsilon_1, \ldots, \varepsilon_{L} $ are independent Rademacher variables, i.e., random variables which take the two values $\pm 1 $ each with probability $ \frac{1}{2} $.
\end{lemma}
\noindent The proof of the Lemma \ref{lemma:smallballmethod} is based on a variant of Mendelson's small-ball method and proceeds in analogy to \cite{koltchinskii2015bounding}. We have deferred a detailed proof to Appendix \ref{section:auxlemma}. In order to apply Lemma \ref{lemma:smallballmethod} we need to estimate the first term of the right-hand side of Lemma \ref{lemma:smallballmethod}. Such an estimate can be derived using the Paley-Zygmund inequality as in \cite{koltchinskii2015bounding}. For the sake of completeness we have included a proof in Appendix \ref{appendix:paleyzygmund}.
\begin{lemma}\label{lemma:BDpaleyzygmundapplied}
	Let $ X \in \mathbb{C}^{K \times N} $ be an arbitrary matrix. Then for all $ \xi > 0$
	\begin{equation*}
	 \sum_{\ell=1}^{L} \mathbb{P} \left( \vert \innerproduct{b_{\ell} c^*_{\ell}, X }_F \vert \ge 2\xi  \right) \ge  \frac{ 9}{32}    \big\vert \left\{ \ell \in \left[L\right] : \ \Vert X^* b_{\ell} \Vert \ge  4 \xi   \right\} \big\vert.
	\end{equation*}
\end{lemma}
In order to use Lemma \ref{lemma:BDpaleyzygmundapplied} we need a lower bound for $\vert \left\{ \ell \in \left[L\right] : \ \Vert X^* b_{\ell} \Vert \ge  \xi   \right\} \vert$. This will be achieved by the next lemma. For the statement of this lemma and its proof we will need to introduce the following notion. We define for any matrix $W \in \mathbb{C}^{K\times N} $ its $\wBnorm{\cdot} $-quasinorm by
\begin{align*}
\wBnorm{W}&:=   \ \underset{ \xi \ge 0}{\sup} \  \xi  \Big\vert  \left\{ \ell \in \lbrack L \rbrack : \Vert W^* b_{\ell} \Vert \ge \xi   \right\}  \Big\vert.
\end{align*}
That is, $ \wBnorm{W} $ is the weak $\ell_1$-norm of the vector $\left( \Vert W^* b_{\ell} \Vert   \right)^L_{\ell=1} $. (For a more detailed discussion of the weak $\ell_1$-norm see, e.g., \cite{grafakos}.) A direct consequence of this interpretation is the inequality (see, e.g., \cite[Proposition 2.10 and Exercise 2.4]{foucart2013mathematical})
\begin{equation}\label{ineq:wnorm}
\wBnorm{W} \le \Bnorm{W} \le \log \left(eL\right) \wBnorm{W}.
\end{equation}

\begin{lemma}\label{lemma:largeentries}
	Let $ Z \in \mathbb{C}^{K \times N} $ such that $ \Vert Z \Vert_F =1 $. Then it holds that
	\begin{equation*}
	\Big\vert \left\{ \ell \in \left[L\right] : \   \Vert Z^* b_{\ell} \Vert     \ge  \frac{\Bnorm{Z}}{L \log \left(eL\right)}    \right\} \Big\vert \ge  \frac{\Bnorm{Z}^2}{\log^2 \left(eL\right)}.
	\end{equation*}
\end{lemma}

\begin{proof}[Proof of Lemma \ref{lemma:largeentries}]
	Choose $\xi^* $ such that
	\begin{equation}\label{ineq:auxiliary7532}
	\wBnorm{Z} = \xi^* \vert \left\{ \ell \in \left[L\right] : \   \Vert Z^* b_{\ell} \Vert     \ge  \xi^*     \right\} \vert.
	\end{equation}	
	As $\vert \left\{ \ell \in \left[L\right] : \   \Vert Z^* b_{\ell} \Vert     \ge  \xi^*     \right\} \vert \le L$ it follows that
	\begin{equation}\label{ineq:auxiliary7531}
	\xi^* \ge \frac{\wBnorm{Z}}{L} \ge \frac{\Bnorm{Z}}{L \log \left(eL\right)} ,
	\end{equation}
	where we also used inequality (\ref{ineq:wnorm}). We observe that
	\begin{align*}
	1&= \Vert Z \Vert^2_F\\
	&=\sum_{\ell=1}^{L} \Vert Z^* b_{\ell} \Vert^2\\
	&\ge {\xi^*}^2 \vert \left\{ \ell \in \left[L\right] : \   \Vert Z^* b_{\ell} \Vert     \ge  \xi^*     \right\} \vert\\
	&\overset{(\ref{ineq:auxiliary7532})}{=} \frac{\wBnorm{Z}^2}{ \vert \left\{ \ell \in \left[L\right] : \   \Vert Z^* b_{\ell} \Vert     \ge  \xi^*     \right\} \vert}\\
	&\overset{(\ref{ineq:wnorm})}{\ge}  \frac{\Bnorm{Z}^2}{ \log^2 \left(eL\right) \vert \left\{ \ell \in \left[L\right] : \   \Vert Z^* b_{\ell} \Vert     \ge  \xi^*     \right\} \vert},
	\end{align*}	
	where for the second equality we used the identity $ \sum_{\ell=1}^{L} b_{\ell} b^*_{\ell} = \Id $. Using inequality (\ref{ineq:auxiliary7531}) and rearranging terms it follows that
	\begin{equation*}
	\Big\vert \left\{ \ell \in \left[L\right] : \   \Vert Z^* b_{\ell} \Vert     \ge  \frac{\Bnorm{Z}}{L \log \left(eL\right)}     \right\} \Big\vert \ge \big\vert \left\{ \ell \in \left[L\right] : \   \Vert Z^* b_{\ell} \Vert     \ge  \xi^*     \right\} \big\vert \ge \frac{\Bnorm{Z}^2}{ \log^2 \left(eL\right)},
	\end{equation*}	
	which completes the proof.

\end{proof}

\noindent The next lemma gives a bound on  $\underset{Z \in \mathcal{E}_{\mu, \delta}}{\inf} \Bnorm{Z} $.
\begin{lemma}\label{lemma:onenormdescentcone}
	It holds that
	\begin{equation*}
	\underset{Z \in \mathcal{E}_{\mu, \delta}}{\inf} \Bnorm{Z} \ge \frac{\delta \sqrt{L}}{\mu}.
	\end{equation*}
\end{lemma}

\begin{proof}
	Let $Z\in \mathcal{E}_{\mu,\delta} $ be arbitrary. By definition of $ \mathcal{E}_{\mu,\delta} $ there is $ h_0 \in \mathcal{H}_{\mu}$ and $m_0 \in \mathbb{C}^N $ such that $Z \in  \dcone{h_0 m^*_0}  $ and such that the inequality 
	\begin{align*}
	\delta&\le  \frac{- \text{Re} \left(  \innerproduct{ Z, h_0 m^*_0 }_F  \right)}{\Vert h_0 m^*_0 \Vert_F} 
	\end{align*}
	holds. It follows that
	\begin{align*}
	\delta&\le   \frac{ - \text{Re} \left( \innerproduct{ Zm_0, h_0  } \right)}{\Vert h_0 m^*_0 \Vert_F}  \\
	& =    \frac{ - \sum_{\ell =1}^L \text{Re} \left( \innerproduct{ Zm_0, b_{\ell}    } \innerproduct{ b_{\ell}, h_0 }  \right)}{\Vert h_0 m^*_0 \Vert_F} \\
	& \le \frac{  \left(   \underset{\ell \in \left[L\right]}{\max} \vert \innerproduct{h_0, b_{\ell}} \vert \right)  \sum_{\ell =1}^L  \vert   \innerproduct{Zm_0, b_{\ell}}  \vert}{\Vert h_0 m^*_0 \Vert_F},
	\end{align*}
	where for the second equality we have used that $ \sum_{\ell =1}^L b_{\ell} b^*_{\ell} = \Id $. Note that for all $\ell \in \left[L\right]$ it holds that 
	\begin{equation*}
	\frac{ \vert \innerproduct{Z m_0, b_{\ell}} \vert}{\Vert m_0 \Vert}  \le \Vert Z^* b_{\ell} \Vert.
	\end{equation*} 
	Hence, by the previous inequality chain it follows that
	\begin{equation*}
	\delta  \le  \frac{   \underset{\ell \in \left[L\right]}{\max} \ \vert \innerproduct{h_0, b_{\ell}} \vert}{\Vert h_0 \Vert}  \sum_{\ell =1}^L  \Vert Z^*b_{\ell} \Vert  \le    \frac{ \mu}{\sqrt{L}} \Vert Z \Vert_{B_1},
	\end{equation*}	
	where in the last inequality we used the definition of $  \mu $ and $ \Vert Z \Vert_{B_1} $. Rearranging terms and taking the infimum over all $Z \in \mathcal{E}_{\mu,\delta} $ yields the desired inequality.
\end{proof}

Having gathered all the necessary ingredients we can state and prove the main lemma of this section.
\begin{lemma}\label{lemma:deconvconicbound}
Let $ \delta >0 $. Assume that
\begin{equation}\label{ineq:assumption1}
L \ge C_1  \left( \frac{\mu}{\delta} \right)^6 \left(K+N\right) \ \log^6 \left(eL\right),
\end{equation}
Then with probability at least  $ 1 -\exp \left( - \frac{L\delta^4 }{C_2 \log^4 \left(eL\right) \mu^4 }  \right) $ it holds that
\begin{equation}\label{ineq:deconvconicbound}
\underset{Z \in \mathcal{E}_{\mu,\delta}}{\inf}  \Vert \mathcal{A} \left(Z\right) \Vert \gtrsim  \frac{\delta^2}{\log^{2} \left( L\right) \mu^{2}}.
\end{equation}
$C_1$ and $C_2$ are absolute constants.
\end{lemma}

\begin{proof}
Our goal is to apply Lemma \ref{lemma:smallballmethod}. In order to apply it we first derive a lower bound for the first term on the right-hand side of inequality (\ref{ineq:intern333333}). For that recall that by Lemma \ref{lemma:onenormdescentcone} it holds that
\begin{equation}\label{intern:bla123}
\underset{Z \in \mathcal{E}_{\mu,\delta}}{\inf} \Vert Z \Vert_{B_1} \ge \frac{\delta\sqrt{L}}{\mu}.
\end{equation}
Thus, for any $Z \in \mathcal{E}_{\mu,\delta} $ we obtain that
\begin{align*}
\Big\vert \left\{ \ell \in \left[L\right]: \ \Vert Z^* b_{\ell} \Vert \ge \ \frac{\delta}{ \mu \sqrt{L} \log \left(eL\right)}    \right\} \Big\vert  & \ge  \Big\vert \left\{ \ell \in \left[L\right]: \ \Vert Z^* b_{\ell} \Vert \ge \ \frac{\Bnorm{Z}}{L \log \left(eL\right)}    \right\} \Big\vert \\
& \ge \frac{\Bnorm{Z}^2}{\log^2 \left(eL\right)}\\
& \ge \frac{\delta^2 L}{\mu^2 \log^2 \left(eL\right)},
\end{align*}
where the first inequality follows from (\ref{intern:bla123}), the second one is due to Lemma \ref{lemma:largeentries}, and the third one follows again from (\ref{intern:bla123}). Hence, by Lemma \ref{lemma:BDpaleyzygmundapplied} applied with $ \xi =  \frac{\delta}{4 \sqrt{L} \ln \left(eL\right) \mu }  $ we finally obtain that
\begin{equation}\label{ineq:marginaltailbound}
\begin{split}
&\underset{Z \in \mathcal{E}_{\mu,\delta}}{\inf} \left(	 \sum_{\ell=1}^{L} \mathbb{P} \left( \vert \innerproduct{b_{\ell} c^*_{\ell}, Z }_F \vert \ge  \frac{\delta}{2 \sqrt{L} \ln \left(eL\right) \mu }  \right)  \right) \\
\ge & \frac{ 9}{32}  \underset{Z \in \mathcal{E}_{\mu,\delta}}{\inf}     \Big\vert \left\{ \ell \in \left[L\right] : \ \Vert Z^* b_{\ell} \Vert \ge  \frac{\delta}{\sqrt{L} \ln \left(eL\right) \mu }   \right\} \Big\vert\\
\ge &   \frac{ 9 \delta^2 L}{32 \mu^2 \log^2 \left(eL\right)}.
\end{split}
\end{equation}
Next, we need an upper bound for the Gaussian width. For that, we first observe that
\begin{align*}
\mathcal{E}_{\mu,\delta}\subset \left( \underset{h_0 \in \mathbb{C}^K, m_0 \in \mathbb{C}^N }{\bigcup}  \dcone{h_0 m^*_0}  \right) \cap \left\{ Z \in \mathbb{C}^{K \times N}: \ \Vert Z \Vert_F =1  \right\}=: \mathcal{E}.
\end{align*} 
The Gaussian width of $\mathcal{E}$ has been bounded in \cite[Lemma 4.1]{kabanava2016stable}, combined with the monotonicity of the Gaussian width their results yields that
\begin{equation}\label{ineq:gaussianwidthbound}
\omega \left( \mathcal{E}_{\mu,\delta}  \right) \le \omega \left( \mathcal{E} \right) \le 2 \sqrt{ \left(K+N\right)}.
\end{equation}
Thus for $\xi= \frac{\delta}{ 4\sqrt{L} \log \left(eL\right) \mu } $ we obtain from Lemma \ref{lemma:smallballmethod} together with (\ref{ineq:marginaltailbound}), (\ref{ineq:gaussianwidthbound}) that with probability at least $ 1 -\exp \left( -2t^2 \right) $ it holds that
\begin{align*}
\underset{Z \in \mathcal{E}_{\mu, \delta}}{\inf}  \Big\vert \left\{ \ell \in    \left[L\right]: \  \vert \langle b_{\ell} c^*_{\ell}, Z \rangle_F \vert \ge  \xi  \right\}   \Big\vert  
\ge   &\frac{9 L\delta^2}{32 \log^2 \left(eL\right) \mu^2 } - \frac{2 \log \left(eL\right) \mu \sqrt{ L \left(K+N \right)} }{\delta} -  t \sqrt{L} \\
\ge &  \frac{9 L\delta^2}{64 \log^2 \left( e L \right) \mu^2} - t \sqrt{L},
\end{align*}
where the second inequality follows from assumption (\ref{ineq:assumption1}), if the constant $C_1>0$ is chosen large enough. Consequently, setting $t= \frac{9 \delta^2 \sqrt{L} }{128 \log^2 \left(eL\right) \mu^2 } $ and recalling that $ \left(\mathcal{A} \left(Z\right)\right) \left(\ell\right) = \innerproduct{b_{\ell}c_{\ell}^*,Z}_F $ we have that with probability at least $ 1 -\exp \left( - \frac{L\delta^4 }{C_2 \log^4 \left(eL\right) \mu^4 }  \right) $ with $C_2$ chosen appropriately
\begin{equation*}
\underset{Z \in \mathcal{E}_{\mu,\delta}}{\inf}  \Big\vert \left\{ \ell \in    \left[L\right]: \ \vert \mathcal{A} \left(Z\right)\left(\ell\right)  \vert \ge  \frac{\delta}{ 4\sqrt{L} \log \left(eL\right) \mu }   \right\}   \Big\vert  \ge  \frac{9 L\delta^2}{128 \log^2 \left(L\right) \mu^2}.
\end{equation*}
Summing up we obtain that with probability at least $ 1 -\exp \left( - \frac{L\delta^4 }{C_2 \log^4 \left(eL\right) \mu^4 }  \right) $
\begin{align*}
\underset{Z \in \mathcal{E}_{\mu,\delta}}{\inf}  \Vert \mathcal{A} \left(Z\right) \Vert &\ge \underset{Z \in \mathcal{E}_{\mu,\delta}}{\inf}  \frac{\delta}{ 4\sqrt{L} \log \left(eL\right) \mu } \sqrt{ \Big\vert \left\{ \ell \in    \left[L\right]: \  \vert \mathcal{A} \left(Z\right)\left(\ell\right) \vert \ge  \frac{\delta}{ 4\sqrt{L} \log \left(eL\right) \mu }   \right\}   \Big\vert }\\
&\gtrsim   \frac{ \delta^{2}}{ \log^{2} \left( L\right) \mu^{2}}.
\end{align*}
This shows the claim.
\end{proof}

\subsection{Proof of Theorem \ref{thm:stabilityBD}}\label{section:proofmain}
As already mentioned in Section \ref{section:proofconic}, in order to control all matrices $Z\in \dcone{h_0 m^*_0}$, which are almost orthogonal to $h_0 m^*_0$, we need the following key lemma.
\begin{lemma}\label{lemma:maximaldescent}
	Let $ h_0 m^*_0 \in \mathbb{C}^{n_1 \times n_2} $ be a rank-1 matrix.  Assume that  $ Z \in  \dcone{h_0 m^*_0}  \setminus \left\{ 0 \right\} $. Then, whenever $ \Vert h_0 m^*_0 +Z \Vert_{\ast} \le \Vert h_0 m^*_0 \Vert_{\ast} $, it holds that 
	\begin{equation*}
	\Vert Z \Vert_F \le -2 \text{Re} \left(  \innerproduct{h_0 m^*_0, \frac{1}{\Vert Z \Vert_F} Z}_F \right).
	\end{equation*}
\end{lemma}
\begin{proof}
	We observe that
	\begin{align*}
	\Vert h_0 m^*_0 \Vert^2_F &=\Vert h_0 m^*_0 \Vert^2_{\ast}\\
	&\ge \Vert h_0 m^*_0 +  Z \Vert^2_{\ast}\\
	&\ge \Vert h_0 m^*_0+  Z \Vert^2_{F}\\
	&= \Vert h_0 m^*_0 \Vert^2_F + \Vert Z \Vert^2_F +  2 \text{Re} \left(  \innerproduct{h_0 m^*_0,Z}_F \right).
	\end{align*}
	Rearranging terms yields the result.
	%
\end{proof}
\noindent Now we have gathered all tools which are needed to prove Theorem \ref{thm:stabilityBD}.
\begin{proof}[Proof of Theorem \ref{thm:stabilityBD}]
Having introduced all necessary tools in the last two sections we can now give a proof of Theorem \ref{thm:stabilityBD}. We set $\delta:= \left(\log eL\right)^{2/3} \mu^{2/3} \alpha^{1/3}$. Throughout the proof we will assume that inequality (\ref{ineq:deconvconicbound}) holds, which by Lemma \ref{lemma:deconvconicbound} holds with probability at least 
\begin{equation*}
1 -\exp \left( - \frac{L\delta^4 }{C_2 \log^4 \left(eL\right) \mu^4 }  \right)  = 1 -\exp \left( - \frac{L \alpha^{4/3} }{C_2 \log^{4/3} \left(eL\right) \mu^{4/3 } }  \right).
\end{equation*}
Let $h_0 \in \mathcal{H}_{\mu} $ and $ m_0 \in \mathbb{C}^N$. Furthermore, let $\hat{X}$ be a minimizer of (\ref{opt:SDP}) and set $Z:= \hat{X} - h_0 m^*_0 $. Note that from the minimality of $\hat{X}$ it follows that $\Vert \hat{X} \Vert_{\ast} \le \Vert h_0 m^*_0 \Vert_{\ast} $. This implies that $ Z \in  \dcone{h_0 m^*_0} $. To prove the lemma it remains to derive an appropriate upper bound on $ \Vert Z \Vert_F $. For that we will distinguish two cases, namely $ \frac{Z}{\Vert Z \Vert_F} \in \mathcal{E}_{\mu,\delta} $ and $ \frac{Z}{\Vert Z \Vert_F} \notin \mathcal{E}_{\mu,\delta} $. If $\frac{Z}{\Vert Z \Vert_F} \in \mathcal{E}_{\mu,\delta} $, it follows from inequality (\ref{ineq:deconvconicbound}) that
\begin{equation}\label{ineq:intern1001}
	\begin{split}
		\Vert Z \Vert_F &\lesssim  \frac{\log^{2} \left( L\right) \mu^{2}}{\delta^2}   \Big\Vert  \mathcal{A} \left(   Z \right)  \Big\Vert\\
		& \le  \frac{\log^{2} \left( L\right) \mu^{2}}{\delta^2}  \left( \Vert \mathcal{A} \left( \hat{X} \right) -y \Vert + \Vert e \Vert  \right) \\
		& \le 2\frac{\log^{2} \left( L\right) \mu^{2}}{\delta^2} \tau\\
		& = 2\frac{  \log^{2/3} \left( L\right)  \mu^{2/3}}{\alpha^{2/3}} \tau,
	\end{split}
 \end{equation} 
where in the second inequality we used the triangle inequality as well as $ Z= \hat{X} -h_0 m^*_0 $ and $ y = \mathcal{A}\left( h_0 m^*_0 \right) + e $. In the third inequality we used that $\hat{X}$ is feasible and $ \Vert e \Vert \le \tau $. If   $\frac{Z}{\Vert Z \Vert_F}  \notin \mathcal{E}_{\mu,\delta} $, it follows directly from the definition of $ \mathcal{E}_{\mu,\delta} $ that $ - \text{Re} \left( \innerproduct{ \frac{ h_0 m^*_0}{\Vert h_0 m^*_0 \Vert_F} , \frac{Z}{\Vert Z \Vert_F} }_F \right) <  \delta $. By Lemma \ref{lemma:maximaldescent} we obtain that  
\begin{equation}\label{ineq:intern1002}
\begin{split}
\Vert Z \Vert_F &\le - 2 \text{Re} \left( \innerproduct{  h_0 m^*_0 , \frac{1}{\Vert Z \Vert_F} Z}_F \right)\\
&< 2\delta \Vert h_0 m^*_0 \Vert_F\\
&< 2  \left( \log L \right)^{2/3} \mu^{2/3} \alpha^{1/3}  \Vert h_0 m^*_0 \Vert_F.
\end{split}
\end{equation}
Combining the estimates (\ref{ineq:intern1001}) and (\ref{ineq:intern1002}) we obtain that 
\begin{align*}
	\Vert \hat{X} - h_0 m^*_0 \Vert_F = \Vert Z \Vert_F &\lesssim  \frac{ \mu^{2/3}  \log^{2/3} L}{\alpha^{2/3}} \max \left\{  \tau  ; \alpha \Vert h_0 m^*_0 \Vert_F  \right\}.
\end{align*}
which completes the proof.
\end{proof}

\section{Outlook}\label{section:outlook}
In this paper we have analyzed two important cases of structured low-rank matrix recovery problems, blind deconvolution and matrix completion, through an inspection of the descent cone of the nuclear norm and its interaction with the measurement operator $ \mathcal{A} $. We have shown that the conic singular value is typically quite small and, consequently, previous analysis approaches cannot give strong recovery guarantees. For the example of blind deconvolution we have presented a new approach based on a refined analysis of the descent cone, showing that the nuclear norm minimization approach is stable against adversarial noise in certain important parameter regimes and allows for uniform recovery guarantees in the presence of noise. In our opinion our results give rise to a number of interesting follow-up questions.
\begin{itemize}
\item \textbf{Stability for small noise-levels: } Until now, our stability result only covers the situation that the noise level $\tau$ is of constant order (up to logarithmic factors). For small $\tau$, Theorem \ref{thm:mainresult}, respectively Theorem \ref{thm:mainresultMC}, put some barriers on what performance can be expected. Nevertheless, it will be interesting to examine the transitional case, that $\tau$ is rather small, even further. For example, while the bad conditioning for small noise levels has been established, it remains open whether one can construct a noise vector $e$ such that the true minimizer behaves like the alternative (but non-minimal) solutions constructed in Theorems \ref{thm:mainresult} and \ref{thm:mainresultMC}. Also the transitional behavior of the minimum conic singular values between very small noise levels (where we established bad conditioning) and larger noise levels (where at least for randomized blind deconvolution, we proved stability) will be an interesting question to study.


\item \textbf{Extension to the rank $r$ case:} 
Understanding nuclear norm recovery for matrix completion under adversarial noise remains an important open problem in the field. While our result established that recovery guarantees for arbitrary noise levels are not feasible, our considerations for the rank one scenario give hope that for sufficiently large noise levels, near optimal guarantees are within reach also for matrices of arbitrary rank.

Similarly, a natural generalization of blind deconvolution is the problem of blind demixing \cite{ling2017blind,jung2018blind}, where one observes a noisy superposition of several convolutions, that is, $y= \sum_{i=1}^{r} w_i \ast x_i +e$. The corresponding low-rank matrix formulation can be interpreted as a rank $r$ version of the randomized blind deconvolution problem.

We expect that a rank $r$ version of Theorem~\ref{thm:stabilityBD} will apply to both these scenarios, which is why we consider this a very promising direction for future research.

\item \textbf{Extension to other low-rank matrix recovery models}: Various other low-rank matrix models also involve incoherence in some way, for example,  robust PCA (\cite{candes2011robust}) and spectral compressed sensing via matrix completion \cite{chen2013spectral}. Also for these problems, recovery results are typically proven via the Golfing Scheme and lead to a seemingly suboptimal noise bound (see, e.g.,  \cite[Section VI]{zhou2010stable}). Can these problems be analyzed with the methods developed in this paper?\\
Moreover, \cite{krahmer2018phase} provided an incoherence based analysis of the phase retrieval problem under random Bernoulli measurements. It will be interesting to analyze this setup with similar methods as in this manuscript. 
\end{itemize}

\subsection*{Acknowledgements}

The authors would like to thank Richard Kueng for discussions related to this topic, which sparked their interest in the problem. Furthermore, the authors want to thank Peter Jung, Marius Junge, and Kiryung Lee for fruitful discussions. This work has been supported by the German Science Foundation (DFG) in the context of the project {\em Bilinear Compressed Sensing} (KR 4512/2-1) as a part of the Priority Program 1798, and the Emmy-Noether Junior Research Group {\em Randomized Sensing and Quantization of Signals and Images} (KR 4512/1-1).

\bibliographystyle{abbrv}
\bibliography{literature}

\appendix

\section{Proof of Lemma \ref{lemma:smallballmethod}}\label{section:auxlemma}
The proof of Lemma \ref{lemma:smallballmethod} will rely on the following two lemmas. The first lemma is a version of Mendelson's small-ball method for non-i.i.d. measurements. In order to state it let $X_1, \ldots, X_{L}$ be independent, matrix-valued random variables defined on a probability space $ \left(\Omega, \mu \right) $. For every measurable, real-valued function $f$ and for every $ \xi > 0 $ we define the quantity
\begin{align*}
Q_{\xi} \left( f\right) & =\sum_{\ell=1}^{L} \mathbb{P} \left( \fmend \ge  \xi  \right).
\end{align*}

\begin{lemma}\label{lemma:smallballmethodgeneral}
	Let $ X_1, \ldots, X_L \in \mathbb{C}^{K \times N} $ be independent random variables and $  \mathcal{F} $ be a set of real-valued functions, which are measurable with respect to $ \left(\Omega, \mu \right) $. Let $t>0$ and $\xi>0$. Then with probability at least $ 1- \exp \left(-2t^2\right) $ it holds that 
	\begin{align*}
	\inff \Big\vert \left\{ \ell \in    \left[L\right]: \  \fmend \ge \xi   \right\}   \Big\vert  \ge   \inff Q_{2\xi} \left( f\right) - \frac{2}{\xi} \mathbb{E} \left[    \supf \sum_{\ell=1}^{L} \varepsilon_{\ell} \fmend  \right] -  t  \sqrt{L},
	\end{align*}
	where $ \varepsilon_1, \ldots, \varepsilon_{L} $ are independent Rademacher variables, i.e., random variables which take the two values $\pm 1 $ each with probability $ \frac{1}{2} $.
\end{lemma}
The proof of Lemma \ref{lemma:smallballmethodgeneral} is exactly analogous as the proof of the original small-ball method \cite{koltchinskii2015bounding}, see Section \ref{section:auxlemma1}. The second auxiliary lemma, proved in Section \ref{section:auxlemma2}, relates the quantity $\mathbb{E} \left[    \supf \sum_{\ell=1}^{L} \varepsilon_{\ell} \fmend  \right]$ in the blind deconvolution framework to the Gaussian width (cf. Definition \ref{def:Gaussianwidth}). 
\begin{lemma}\label{lemma:meanempiricalwidth}
	Let $ \mathcal{E} \subset \mathbb{C}^{K \times N}$. Then it holds that
	\begin{equation*}
	\mathbb{E} \left[ \underset{X \in \mathcal{E}}{\sup} \ \real \left( \sum_{\ell=1}^{L} b^*_{\ell} X c_{\ell} \right)  \right] = \omega \left( \mathcal{E} \right).
	\end{equation*}
\end{lemma}
With these lemmas we can now prove Lemma \ref{lemma:smallballmethod}.
\begin{proof}[Proof of Lemma \ref{lemma:smallballmethod}]
	Set $X_{\ell} := b_{\ell} c^*_{\ell} $ for all $ \ell \in \left[L\right] $ and define 
	\begin{equation*}
	\mathcal{F} := \left\{ \vert \innerproduct{ M, \cdot }_F \vert  : M \in \mathcal{E}    \right\} .
	\end{equation*}
	Then, by a direct application of Lemma \ref{lemma:smallballmethodgeneral} we obtain that with probability at least $ 1-\exp \left(-2t^2\right) $ it holds that
	\begin{equation}\label{ineq:intern4444}
	\begin{split}
	&\underset{M\in \mathcal{E}}{\inf} \Big\vert \left\{ \ell \in    \left[L\right]: \  \vert \innerproduct{ M, b_{\ell} c^*_{\ell} }_F \vert \ge \xi   \right\}   \Big\vert \\ \ge &  \underset{M \in \mathcal{E}}{\inf} \left( \sum_{\ell=1}^{L} \mathbb{P} \left( \vert \langle b_{\ell} c^*_{\ell}, M \rangle_F \vert \ge  2\xi  \right)   \right) 
	- \frac{2}{\xi} \mathbb{E} \left[   \underset{M \in \mathcal{E}}{\sup} \sum_{\ell=1}^{L} \varepsilon_{\ell} \vert \innerproduct{  b_{\ell} c^*_{\ell}, M }_F \vert  \right] -  t  \sqrt{L}.
	\end{split}
	\end{equation}
	To bound the second summand, we observe that
	\begin{equation}\label{ineq:intern4445}
	\begin{split}
	&\mathbb{E} \left[   \underset{M \in \mathcal{E}}{\sup} \sum_{\ell=1}^{L} \varepsilon_{\ell} \vert \innerproduct{  b_{\ell} c^*_{\ell}, M }_F \vert  \right]\\
	\le & \mathbb{E} \left[   \underset{M \in \mathcal{E}}{\sup} \sum_{\ell=1}^{L} \varepsilon_{\ell} \vert \text{Re} \left( \innerproduct{  b_{\ell} c^*_{\ell}, M }_F \right) \vert  \right] + \mathbb{E} \left[   \underset{M \in \mathcal{E}}{\sup} \sum_{\ell=1}^{L} \varepsilon_{\ell} \vert \text{Im} \left( \innerproduct{  b_{\ell} c^*_{\ell}, M }_F \right) \vert  \right] \\
	= &  2\ \mathbb{E} \left[   \underset{M \in \mathcal{E}}{\sup} \sum_{\ell=1}^{L} \varepsilon_{\ell} \vert \text{Re} \left( \innerproduct{  b_{\ell} c^*_{\ell}, M }_F \right) \vert  \right]\\
	= &  2 \ \mathbb{E} \left[   \underset{M \in \mathcal{E}}{\sup} \sum_{\ell=1}^{L} \varepsilon_{\ell} \text{Re} \left( \innerproduct{  b_{\ell} c^*_{\ell}, M }_F \right)   \right]\\
	= & 2 \omega \left( \mathcal{E} \right)
	\end{split}
	\end{equation}
	where in the third line we used that $ \text{Re} \left( \innerproduct{  b_{\ell} c^*_{\ell}, X }_F \right)$  and $\text{Im} \left( \innerproduct{  b_{\ell} c^*_{\ell}, X }_F \right) $ have the same distribution. The fourth line follows from the symmetry of the set $ \mathcal{E} $ and the last line is due to Lemma \ref{lemma:meanempiricalwidth}. Combining (\ref{ineq:intern4444}) and (\ref{ineq:intern4445}) finishes the proof.
\end{proof}

\subsection{Proof of Lemma \ref{lemma:smallballmethodgeneral}}\label{section:auxlemma1}
We directly trace the steps of the proof of Theorem 1.5 in \cite{koltchinskii2015bounding}. In the following $ \mathbbm{1}_A $ denotes the indicator function, which takes the value $1$, if the event $A$ occurs and the value $0$ otherwise. Note that
\begin{align*}
\xi \  \Big\vert \left\{ \ell \in    \left[L\right]: \  \fmend \ge \xi  \right\}   \Big\vert =   \xi  \sum_{\ell=1}^{L} \mathbbm{1}_{\left\{ \fmend  \ge \xi  \right\}}.
\end{align*}
Taking the infimum we observe that by the definition of $ Q_{2\xi} $
\begin{equation}\label{eq:estimate1}
\begin{split}
& \xi \inff  \  \Big\vert \left\{ \ell \in    \left[L\right]: \  \fmend \ge \xi \right\} \Big\vert\\
\ge  & \xi  \underset{f \in \mathcal{F}}{\inf} Q_{2\xi} \left( f \right)   - \xi \ \supf     \sum_{\ell=1}^{L} \left( \mathbb{P} \left( \fmend \ge 2  \xi  \right)   - \mathbbm{1}_{ \left\{ \fmend \ge \xi  \right\} }\right).
\end{split}
\end{equation}
The bounded difference inequality (see, for example, \cite{boucheron2013concentration}) implies that with probability at least $1 - \exp \left( -2t^2 \right) $ it holds that
\begin{equation}\label{eq:estimate2}
\begin{split}
& \supf   \sum_{\ell=1}^{L} \left( \mathbb{P} \left( \fmend \ge 2 \xi  \right)    - \mathbbm{1}_{\left\{ \fmend \ge  \xi  \right\} } \right)  \\
\le & \mathbb{E} \left[  \supf  \sum_{\ell=1}^{L} \left( \mathbb{P} \left( \fmend \ge 2 \xi  \right)     - \mathbbm{1}_{\left\{ \fmend \ge \xi  \right\} }  \right)   \right]   + t \sqrt{L}\\
=&\mathbb{E} \left[  \supf     \sum_{\ell=1}^{L} \left(  \mathbb{E} \left[ \mathbbm{1}_{\left\{ \fmend \ge 2 \xi  \right\} } \right]    - \mathbbm{1}_{\left\{ \fmend \ge \xi  \right\} } \right)    \right] + t \sqrt{L}
\end{split}
\end{equation}
To deal with the expectation we will use the function $ \Psi_{\xi}: [0, + \infty ) \longrightarrow \mathbb{R} $ defined by
\begin{equation*}
\Psi_{\xi} \left( u \right) = \begin{cases}
0  \quad & 0 \le u \le \xi \\
\frac{1}{\xi} \left(u - \xi \right) \quad & \xi \le u \le 2 \xi\\
1  \quad & u \ge 2 \xi
\end{cases}.
\end{equation*}
We observe that $\Psi_{\xi}$ is Lipschitz continuous with Lipschitz constant $1/\xi$. Furthermore, for all $u \in [0,+\infty )$ it holds that $ \mathbbm{1}_{\left\{ u \ge  2\xi \right\}}  \le \Psi_{\xi} \left( u \right) \le \mathbbm{1}_{\left\{ u \ge \xi \right\}} $. Combining this monotonicity relation with Gine-Zinn symmetrization (see, e.g., \cite[Lemma 2.3.1]{van1996weak}) and the Rademacher comparison principle for Lipschitz continuous functions (see, e.g., \cite[Equation (4.20)]{ledoux2013probability}), we obtain that
\begin{align*}
& \mathbb{E} \left[ \supf    \sum_{\ell=1}^{L} \left( \mathbb{E} \left[\mathbbm{1}_{\left\{ \fmend \ge 2 \xi  \right\}}  \right]  - \mathbbm{1}_{ \left\{ \fmend \ge \xi  \right\}}   \right)  \right] \\
\le &  \mathbb{E} \left[ \supf  \sum_{\ell=1}^{L} \left( \mathbb{E} \left[ \Psi_{\xi} \left( \fmend \right)  \right]  - \Psi_{\xi} \left( \fmend \right)  \right)   \right]\\
\le & 2 \ \mathbb{E} \left[   \supf \sum_{\ell=1}^{L} \varepsilon_{\ell}  \Psi_{\xi} \left(  \fmend \right)     \right]\\
\le & \frac{2}{\xi} \ \mathbb{E} \left[  \supf \sum_{\ell=1}^{L} \varepsilon_{\ell} \fmend   \right].
\end{align*}
Together with the inequality chains (\ref{eq:estimate1}) and (\ref{eq:estimate2}), this completes the proof.
\qed

\subsection{Proof of Lemma \ref{lemma:meanempiricalwidth}}\label{section:auxlemma2}
First, we observe that
\begin{equation*}
\mathbb{E} \left[ \underset{X \in \mathcal{E}}{\sup} \ \real \left( \sum_{\ell=1}^{L} b^*_{\ell} X c_{\ell} \right)  \right] = 	\mathbb{E} \left[ \underset{X \in \mathcal{E}}{\sup} \ \real \left( \innerproduct{X, \sum_{\ell=1}^{L} b_{\ell} c^*_{\ell} }_F \right)  \right]
\end{equation*}
Note that due to the definition of $ \mathcal{\omega} \left( \mathcal{E} \right) $ in order to finish the proof it is enough to show that the entries of the matrix $ X= \sum_{\ell =1}^L b_{\ell} c^*_{\ell} $ are independent and identically distributed with distribution $ \mathcal{CN} \left(0, 1  \right) $. For that, let $\left(i,j\right) \in \left[K\right] \times \left[N\right] $ and compute that
\begin{align*}
\mathbb{E} \left[ \vert e^*_i \left( \sum_{\ell=1}^{L} b_{\ell} c^*_{\ell} \right) e_j  \vert^2  \right] &= \sum_{\ell =1}^L e^*_i b_{\ell} b^*_{\ell} e_i \mathbb{E}  \left[ \vert c^*_{\ell} e_j \vert^2 \right]^2 = \sum_{\ell =1}^L e^*_i b_{\ell} b^*_{\ell} e_i =1.
\end{align*}
This implies that $   e^*_i \left( \sum_{\ell=1}^{L} b_{\ell} c^*_{\ell} \right) e_j  \in \mathcal{CN} \left(0,1\right) $. It remains to show that the individual entries of the matrix $ \sum_{\ell =1}^m b_{\ell} c^*_{\ell} $ are independent. For that, we set 
\begin{equation*}
X_{i,j}:= \left( \sum_{\ell =1}^m b_{\ell} c^*_{\ell} \right)_{i,j} = e^*_i \left( \sum_{\ell=1}^{L} b_{\ell} c^*_{\ell} \right) e_j.
\end{equation*}
Now let $ \left(i,j\right), \left(i',j'\right) \in \left[K\right] \times \left[N\right]  $ such that $ \left(i,j\right) \ne \left(i',j'\right) $. Our goal is to show that $\mathbb{E} \left[  X_{i,j} \overline{X}_{i',j'}  \right] = 0$.  
If $j \ne j' $ this follows immediately from the observation that $c^*_{\ell} e_j$ and $c^*_{\ell} e_{j'} $ are independent for all $\ell \in \left[L\right] $. Now assume that $j=j'$. Then we can compute that
\begin{equation*}
\mathbb{E} \left[  X_{i,j} \overline{X}_{i',j'}  \right] = \sum_{\ell=1}^{L} e^*_{i} b_{\ell} b^*_{\ell} e_{i'} \vert c^*_{\ell} e_j \vert^2  = \sum_{\ell=1}^{L} e^*_{i} b_{\ell} b^*_{\ell} e_{i'} =0.
\end{equation*}
Hence, we have shown that all entries of the matrix $X $ are uncorrelated. As the entries of $X$ are jointly Gaussian this implies that they are independent, which completes the proof.

\qed

\section{Proof of Lemma \ref{lemma:BDpaleyzygmundapplied}}\label{appendix:paleyzygmund}

\begin{proof}[Proof of Lemma \ref{lemma:BDpaleyzygmundapplied}]
	Let $ \ell \in \left[L \right] $ such that $ \Vert  X^* b_{\ell} \Vert \ge 4 \xi  $. Using the Paley-Zygmund inequality (see, e.g., \cite{decoupling}) we obtain that
	\begin{align*}
	\mathbb{P} \left( \vert \innerproduct{b_{\ell} c^*_{\ell}, X }_F \vert \ge 2 \xi  \right) & \ge  \mathbb{P} \left( \vert \innerproduct{b_{\ell} c^*_{\ell}, X }_F \vert \ge \frac{1}{2}  \Vert X^* b_{\ell} \Vert   \right)\\
	&\ge  \frac{ \left(  \E \left[ \vert \innerproduct{b_{\ell} c^*_{\ell}, X }_F \vert^2   \right]    - \frac{1}{4} \Vert X^* b_{\ell} \Vert^2 \right)^2 }{\E \left[ \vert \innerproduct{b_{\ell} c^*_{\ell}, X }_F \vert^4 \right]}\\
	& = \frac{  \left( \Vert X^* b_{\ell} \Vert^2  - \frac{1}{4} \Vert X^* b_{\ell} \Vert^2  \right)^2  }{  2 \Vert X^* b_{\ell} \Vert^4 }  = \frac{9}{32}.
	\end{align*}
	(We used that $ \mathbb{E} \vert \innerproduct{ b_{\ell} c_{\ell}^*, X }_F \vert^2 =  \Vert X^* b_{\ell} \Vert^2 $ and $ \mathbb{E} \vert \innerproduct{ b_{\ell} c_{\ell}^*, X }_F \vert^4 = 2 \Vert X^* b_{\ell} \Vert^4   $, which is due to $ \innerproduct{ b_{\ell} c_{\ell}^*, X }_F \sim \mathcal{CN} \left( 0, \Vert X^* b_{\ell} \Vert  \right) $.) 
	Summing over all $ \ell \in \left[ L \right] $ such that $ \Vert X^* b_{\ell} \Vert \ge 4  \xi  $ yields the claim.
\end{proof}

\newpage

\end{document}